\documentclass[conference, onecolumn, 11pt]{IEEEtran}
\IEEEoverridecommandlockouts

\usepackage{cite}
\usepackage{amsmath,amssymb,amsfonts}
\usepackage{amsthm}
\newtheorem{proposition}{Proposition}
\usepackage{graphicx}
\usepackage{textcomp}
\usepackage{xcolor}
\usepackage{placeins}

\usepackage[english]{babel}

\usepackage[letterpaper,top=2cm,bottom=2cm,left=3cm,right=3cm,marginparwidth=1.75cm]{geometry}

\usepackage{amsmath}
\usepackage{amssymb}
\usepackage{graphicx}
\usepackage[colorlinks=true, allcolors=blue]{hyperref}
\usepackage{amsthm}
\usepackage{algorithm}
\usepackage{algpseudocode}
\newtheorem{definition}{Definition}
\newtheorem{lemma}{Lemma}
\newtheorem{theorem}{Theorem}

\renewcommand{\thesection}{\arabic{section}}
\renewcommand{\thesubsection}{(\alph{subsection})}
\renewcommand{\theequation}{\thesection.\arabic{equation}}

\newcommand{\tRSi}{\text{iRS}}
\newcommand{\tiRS}{\text{iRS}}
\newcommand{\tRSo}{\text{oRS}}
\newcommand{\toRS}{\text{oRS}}

\newcommand{\tRSmRCi}{\text{iRS-mRC}}
\newcommand{\tRSmRCo}{\text{oRS-mRC}}

\newcommand{\tMmFairi}{\text{iMmFair}}
\newcommand{\tMmFairo}{\text{oMmFair}}

\newcommand{\tNP}{\text{NP}}
\newcommand{\tFPTAS}{\text{FPTAS}}
\newcommand{\ee}{\mathbf{e}}
\newcommand{\bF}{\textbf{F}}
\newcommand{\bW}{\textbf{W}}
\newcommand{\bT}{\textbf{T}}
\newcommand{\cD}{\mathcal{D}}
\newcommand{\bETA}{\boldsymbol{\eta}}
\newcommand{\bX}{{\bf X}}
\newcommand{\bY}{{\bf Y}}
\newcommand{\bR}{{\bf R}}

\newcommand{\cI}{\mathcal{I}}

\newcommand{\cQ}{\mathcal{Q}}
\newcommand{\cU}{\mathcal{U}}
\newcommand{\tLB}{\text{LB}}
\newcommand{\tUB}{\text{UB}}
\newcommand{\cP}{\mathcal{P}}
\newcommand{\tQCAST}{\text{Q-CAST}}
\numberwithin{equation}{section}
\usepackage{booktabs}
\usepackage{mathtools}
\usepackage{caption}
\hypersetup{hidelinks}

\def\BibTeX{{\rm B\kern-.05em{\sc i\kern-.025em b}\kern-.08em
    T\kern-.1667em\lower.7ex\hbox{E}\kern-.125emX}}
\begin{document}
\title{Multi-Pair Fidelity-Aware Rate Allocation in a Quantum Network:
Approximation Schemes~\thanks{Zunzheng Zhang, Xuanli Lin, and Guoliang Xue
are affiliated with Arizona State University, Tempe, Arizona, USA.
Emails: \{zzhan621, xlin54, xue\}@asu.edu.
Zhaofeng Zhang is affiliated with Delaware State University, Dover, Delaware, USA.
Email: zzhang@desu.edu.
Nageswara S. V. Rao is affiliated with Oak Ridge National Laboratory, Oak Ridge, Tennessee, USA.
Email: raons@ornl.gov.
This research is sponsored
in part by PiQSci project of Advanced Scientific Computing Research program,
U.S. Department of Energy, and is performed at Oak Ridge National Laboratory
managed by UT-Battelle, LLC for U.S. Department of Energy under Contract
No. DE-AC05-00OR22725.
{This manuscript has been co-authored by UT-Battelle, LLC, under contract
DE-AC05-00OR22725 with the US Department of Energy (DOE).
The US government retains and the publisher, by accepting the article for
publication, acknowledges that the US government retains a nonexclusive,
paid-up, irrevocable, worldwide license to publish or reproduce the published
form of this manuscript, or allow others to do so, for US government purposes.
DOE will provide public access to these results of federally sponsored
research in accordance with the DOE Public Access Plan
(http://energy.gov/downloads/doe-public-access-plan).}}
}
\IEEEaftertitletext{\vspace{-1.0\baselineskip}}
\author{Zunzheng Zhang, Xuanli Lin, Zhaofeng Zhang, Nageswara S. V. Rao, Guoliang Xue}
%
%
%
%
\maketitle
\pagestyle{plain}
\thispagestyle{plain}
\begin{abstract}
Entanglement distribution in quantum networks must jointly account for limited
link capacities, probabilistic entanglement swapping, and heterogeneous link fidelities.
In this paper, we study multi-pair fidelity-aware rate allocation in quantum networks. 
We formulate three rate-allocation problems: rate sum, rate sum subject to minimum-rate constraints, and max-min fairness.
Prior work has studied a special case of the rate sum problem, where
all links have identical fidelity.
This special case admits a polynomial-time algorithm.
We prove that all three problems are $\tNP$-hard.
We then study optimization versions of these problems which maximize the
minimum end-to-end fidelity subject to throughput or fairness requirements.
We present fully polynomial-time approximation schemes (FPTAS)
for solving these optimization problems.
%
%
%
%
Experiments on randomly generated networks demonstrate the computational
effectiveness of the proposed schemes.
\end{abstract}


\section{Introduction}
\label{Intro}
\noindent 
Quantum networks enable the distribution of entanglement between remote quantum
nodes, which is a fundamental resource for quantum communication, distributed
quantum computation, clock synchronization, and other quantum information
processing tasks~\cite{van2014quantum, van2022quantum, komar2014quantum}.
In such networks, intermediate quantum repeaters extend entanglement over long
distances by generating elementary EPR pairs over physical links and performing
entanglement swapping operations~\cite{azuma2023quantum}.
A central objective in quantum network routing is to allocate limited
entanglement generation resources among multiple source-destination pairs.
Unlike classical multicommodity flow routing, both throughput and fidelity
depend on the selected paths: probabilistic entanglement swapping reduces the
achievable end-to-end rate as the number of swaps increases~\cite{briegel1998quantum},
while the Werner-state model imposes end-to-end fidelity constraints that
restrict the set of feasible routes.
Together with limited link capacities and multiple simultaneous requests,
these features make fidelity-constrained quantum network routing a challenging
optimization problem.

Different applications may favor different rate-allocation objectives
for multiple source-destination pairs.
Considering total throughput promotes network utilization but may leave some
pairs with little or no service.
Adding per-pair minimum-rate constraints provides baseline service guarantees,
whereas max-min fairness gives stronger protection to the least-served pairs.


In this paper, we study multi-pair fidelity-aware rate allocation in quantum networks.
We formulate three rate-allocation problems: rate sum, rate sum subject to minimum-rate
constraints, and max-min fairness.
Prior work~\cite{chakraborty2020entanglement} has studied a special case of the
rate sum problem.
This special case admits a polynomial-time algorithm.
However, the polynomial-time algorithm in~\cite{chakraborty2020entanglement}
heavily depends on the assumption that all links in the network have identical fidelity.
We prove that all three problems are NP-hard.
We then study optimization versions of these problems which maximize the
minimum end-to-end fidelity subject to throughput or fairness requirements.
By transforming the multiplicative Werner parameter into an additive loss metric
and building on scaling techniques for multi-constrained routing~\cite{misra2009polynomial,lorenz2001simple}, we develop fully polynomial-time approximation schemes (FPTAS) for
solving these optimization problems.

%

%
%

Our main contributions are summarized as follows:
\begin{itemize}
\item
We formulate multi-pair fidelity-aware rate allocation in quantum networks and
define three problems: rate sum, rate sum subject to minimum-rate constraints,
and max-min fairness.
%
%

\item
We prove that all three problems are $\tNP$-hard, even for the single-pair case.

\item
We then study optimization versions of these problems, and design $\tFPTAS$
for all these optimization problems.

\item
We provide evaluation results to validate the effectiveness of our schemes.
\end{itemize}

The rest of this paper is organized as follows.
$\S$\ref{sec:Related} reviews related work.
$\S$\ref{sec3:model} presents the network model and formally defines the decision problems.
$\S$\ref{sec:opt_def} defines the optimization versions of these decision problems.
$\S$\ref{sec:integer_loss_exact} develops an exact formulation for the integer
version of the rate sum problem.
$\S$\ref{sec:scaling_based_approximation} proposes the scaling-based FPTAS for the rate sum problem.
$\S$\ref{sec:throughput_fairness_algorithms} proposes the scaling-based FPTAS for the rate sum problem
subject to minimum-rate constraints.
$\S$\ref{sec:FPTAS_MmFair} proposes the scaling-based FPTAS for the max-min fairness problem.
$\S$\ref{esc:evaluation} presents the performance evaluation.
Finally, $\S$\ref{sec:conclusion} concludes the paper.
\section{Related Work}
\label{sec:Related}
\noindent
Entanglement routing has been studied extensively~\cite{abane2025entanglement}.
Existing studies can be broadly classified into 
throughput-oriented routing,
fidelity-aware routing and purification,
flow-based resource allocation,
and
fairness-aware allocation.

Some studies focus on routing and resource allocation for a single
source-destination pair, including entanglement-rate analysis,
fidelity-rate optimization, and generalized-flow
maximization~\cite{caleffi2017optimal,gu2024fendi,vardoyan2024bipartite}.
Other studies consider multiple concurrent connection requests.
Pant et al.~\cite{pant2019routing} show that multipath routing can improve
entanglement-distribution rates.
Shi and Qian~\cite{shi2020concurrent}
propose $\tQCAST$ for concurrent entanglement routing over probabilistic links
with path recovery, while Zeng et al.~\cite{zeng2023entanglement} consider both the number of users and the resulting throughput.
Gu et al.~\cite{gu2023esdi} study entanglement scheduling for multiple source-destination pairs.

Fidelity-aware routing and resource allocation have also been studied
through flow-based models and purification-enabled routing schemes.
Zhao et al.~\cite{zhao2022e2e} jointly optimize routing and purification
to maximize the number of connections that meet a prescribed fidelity threshold.
Zhang et al.~\cite{zhang2025link} study link configuration for
fidelity-constrained entanglement routing.
Li et al.~\cite{li2022fidelity} develop purification-enabled, fidelity-guaranteed
routing algorithms together with a greedy multi-pair allocation scheme.
Gu et al.~\cite{gu2025cost} develop an approximation scheme for cost-aware
entanglement distribution with purification.
Flow-based formulations are also considered
in~\cite{zhao2021redundant,dai2020optimal, chakraborty2020entanglement}.

Among these studies, the work of Chakraborty et al.\cite{chakraborty2020entanglement}
is the most closely related to our rate sum formulation.
They formulate entanglement distribution as a multicommodity-flow optimization
problem and provide a polynomial-time solution under the assumption that all
physical links have identical fidelity.
In contrast, we allow heterogeneous link fidelities and show that the resulting
problem is NP-hard even for a single source-destination pair.

Fairness in entanglement routing has also received attention. Li et al.~\cite{li2021effective}
study fidelity-constrained routing together with fair link-level capacity allocation.
They also state that a related multicommodity-flow routing problem is NP-hard,
but do not provide a formal proof.
Wang et al.~\cite{wang2024fair} address max-min fairness through offline predictive
allocation and online resource adjustment, but do not impose fidelity constraints.
%
%

Our approximation schemes build on classical constrained-routing techniques~\cite{nagi1998, wang1996quality}.
Lorenz and Raz~\cite{lorenz2001simple} develop an approximation scheme for
the restricted shortest-path problem, while Xue et al.~\cite{xue2008polynomial} and Misra et al.~\cite{misra2009polynomial}
 study polynomial-time approximation algorithms
for multi-constrained and multi-path routing.
Max-min fairness has been studied extensively in classical networks~\cite{joe2013multiresource,nace2008max, lan2010axiomatic}, and we adapt these ideas to quantum networks.

%
%
\section{System Model and Problem Definitions}
\label{sec3:model}
\noindent
In this section, we present the quantum network model and define the multi-pair
fidelity-aware rate allocation problems.
%
%
\subsection{Quantum Network Model}
\noindent
We model a quantum network as a directed graph $G=(V, E, c, w, q)$, where
$V$ is the set of $n$ quantum nodes, including end nodes and quantum repeaters,
and $E$ is the set of $m$ directed edges.
Each edge represents a physical quantum communication link
(or a direction of a physical link).
We assume that the underlying undirected graph of $G$ is connected.
For an edge $e=(u,v)\in E$, $c(e)$ denotes the maximum
elementary EPR-pair generation rate over the corresponding link~\cite{chakraborty2020entanglement},
and the Werner parameter of $e$ is denoted by
$w(e) = w(u,v) \in (0,1]$.
We also assume that swapping operations are independent
and each succeeds with the same probability
$q\in(0,1]$.

A path from a source node $s$ to a destination node $t$ is denoted by
$s$$\to$$u_1$$\to$$u_2$$\to$$\cdots$$\to$$u_{|\pi|-1}$$\to$$t$,
where $|\pi|$ is the number of edges on the path.
In Fig.~\ref{fig:01}, we illustrate four paths
$\pi_1: a \to c$,
$\pi_2: a \to u \to v \to c$,
$\pi_3 : b \to u \to v \to d$,
and
$\pi_4 : b \to d$.

%
\begin{figure}[htbp]
\centering
\includegraphics[width=0.5\linewidth]{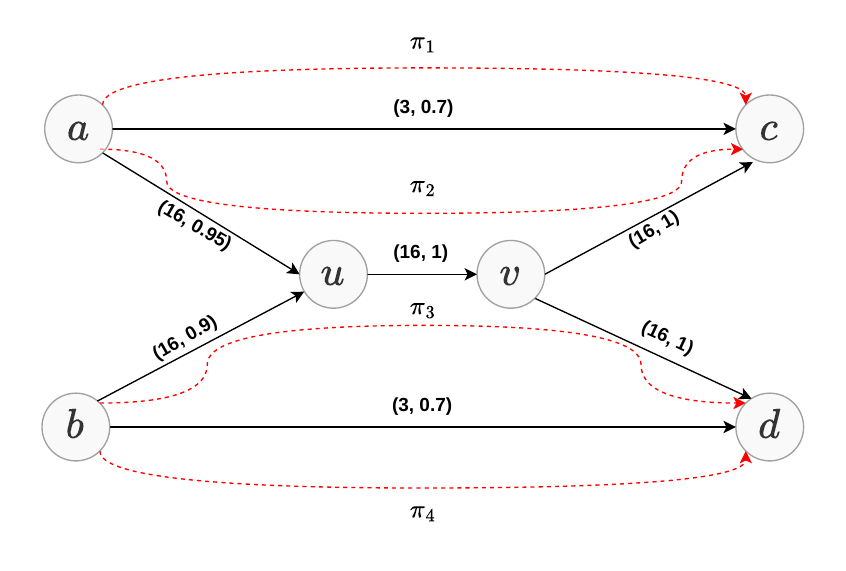}
\caption{Graph $G=(V,E,c, w, q)$.
The 2-tuple label next to directed edge $(u,v)$ denotes $(c(u,v), w(u,v))$.
We assume $q = 0.5$ for every swapping operation.
}
\label{fig:01}
\end{figure}
%

%
Under the Werner-state model~\cite{werner1989quantum, briegel1998quantum}, the end-to-end Werner parameter
of a path $\pi$ is the product of the Werner parameters of the edges
on the path:
\begin{align}
\label{eqn:new3.1}
W(\pi)=\prod_{e\in \pi} w(e).
\end{align}
%
%

Therefore, the end-to-end fidelity of an EPR pair distributed along path $\pi$ is
\begin{align}
\label{eqn:new3.2}
F(\pi)  = \frac{1}{4} + \frac{3}{4} \prod_{e\in \pi} w(e).
\end{align}
%
%

For each edge \(e\in E\), we define its fidelity-loss value as
\begin{align}
\label{eqn:new3.3}
\ell(e) = -\ln w(e).
\end{align}

The fidelity-loss representation is introduced mainly for algorithmic purposes: it transforms the multiplicative Werner-parameter constraint into an equivalent additive path constraint.
For a path $\pi$, define
\begin{align}
\label{eqn:new3.4}
\ell(\pi) = \sum_{e \in \pi} \ell(e).
\end{align}

Using Fig.~\ref{fig:01} as an example again, we have
$W(\pi_1) = 0.7$,
$W(\pi_2) = 0.95 \times 1 \times 1 = 0.95$.
Therefore
$F(\pi_1) = (3 \times 0.7 + 1)/4 = 0.775 $,
$F(\pi_2) = (3 \times 0.95 + 1)/4 = 0.9625$;
$\ell(\pi_1) = -\ln 0.7 $,
$\ell(\pi_2) = - \ln 0.95$.

\subsection{Connection Demands and Rate Allocation}
\noindent
We consider a set of $K$ source-destination demands, also referred to as 
connection requests or commodities,
denoted by
$\mathcal{D}=\{(s_1,t_1),(s_2,t_2),\ldots, \allowbreak (s_K,t_K)\}$,
where $s_k, t_k \in V$ are the source and destination nodes for
commodity $k$.
We assume that $s_k \neq t_k$.
In this paper, we will use the notation $[K] \triangleq \{1, 2, \ldots, K\}$.
%
%
For each demand $k \in [K]$, we need to find a rate allocation
$\cP_k = \{ (\pi_{k, 1}, r_{k, 1}),  (\pi_{k, 2}, r_{k, 2}), \ldots,  (\pi_{k, |\cP_k|}, r_{k, |\cP_k|}) \}$,
where $\pi_{k, i}$ is a simple $s_k$-$t_k$ path
and $r_{k, i} \ge 0$ is the delivered rate along $\pi_{k, i}$,
$i=1, 2, \ldots, |\cP_k|$.
We denote $\cP = ( \cP_1, \cP_2, \ldots, \cP_K )$.
Note that it suffices to consider only simple paths.
If an $s_k$-$t_k$ walk contains
a cycle, removing the cycle cannot increase its fidelity loss.
Moreover, since $q\le 1$, removing the cycle reduces the number of
swapping operations and therefore cannot increase the elementary-pair
generation rate required to support the same delivered rate.
Hence, any feasible allocation can be transformed into one using only
simple paths.
By allowing $\cP_k = \emptyset$, we can restrict every path to have
a positive rate, rather than a non-negative rate.
However, for ease of formulation of the optimization problems, we allow paths
with zero rate.

We say $\cP$ satisfies the link-capacity constraints if
%
\begin{align}
\label{eqn:new3.5}
\sum_{k=1}^{K} \sum_{\substack{1 \le i \le |\cP_k| \\
e \in \pi_{k, i}}}
\frac{r_{k, i}}{q^{|\pi_{k, i}| - 1}}
\le c(e),~~
\forall e\in E.
\end{align}
%

The quantity $r_{k, i}/q^{|\pi_{k, i}| - 1}$ in~\eqref{eqn:new3.5} represents
the elementary-pair generation rate required on each edge of path $\pi_{k,i}$
to achieve a delivered end-to-end rate $r_{k,i}$.
The quantity $q^{|\pi_{k,i}|-1}$ in~\eqref{eqn:new3.5} arises from the 
probabilistic nature of entanglement swapping. 
Since each swapping operation succeeds with probability $q$,
delivering rate $r_\pi$ along a path with $|\pi_{k,i}|-1$ swapping
operations requires $r_\pi/q^{|\pi_{k,i}|-1}$ elementary-pair generation
rate on each edge of the path.
%
%
%
%
Constraint~\eqref{eqn:new3.5} guarantees that for each edge, the aggregated elementary-pair generation rate does not
exceed its capacity.
%
We will simplify $(\pi_{k, i}, r_{k, i}) \in \cP_k$ as $(\pi, r_\pi) \in \cP_k$
if the exact meaning is clear from the context.

In Fig.~\ref{fig:01}, we consider a demand set
$\cD = \{(s_1, t_1), (s_2, t_2)\}$ where
$s_1 = a$, $t_1 = c$, $s_2 = b$, $t_2 = d$.
For illustration, consider a possible rate allocation $\cP=(\cP_1, \cP_2)$, where
$\cP_1 = \{(\pi_1, r_{\pi_1}), (\pi_2, r_{\pi_2}) \} $
and 
$\cP_2 = \{(\pi_3, r_{\pi_3}), (\pi_4, r_{\pi_4}) \} $.
By (\ref{eqn:new3.5}), the link-capacity constraints become:
%
\begin{subequations}
\label{eqn:new3.6}
\begin{alignat}{2}
r_{\pi_1}
&\le 3,
&\qquad&
\text{capacity constraint on edge $(a,c)$,}
\label{eqn:new3.6a}
\\
\frac{r_{\pi_2}}{0.5^2}
&\le 16,
&&
\text{capacity constraint on edge $(a,u)$,}
\label{eqn:new3.6b}
\\
\frac{r_{\pi_3}}{0.5^2}
&\le 16,
&&
\text{capacity constraint on edge $(b,u)$,}
\label{eqn:new3.6c}
\\
\frac{r_{\pi_2}}{0.5^2}
+
\frac{r_{\pi_3}}{0.5^2}
&\le 16,
&&
\text{capacity constraint on edge $(u,v)$,}
\label{eqn:new3.6d}
\\
\frac{r_{\pi_2}}{0.5^2}
&\le 16,
&&
\text{capacity constraint on edge $(v,c)$,}
\label{eqn:new3.6e}
\\
\frac{r_{\pi_3}}{0.5^2}
&\le 16,
&&
\text{capacity constraint on edge $(v,d)$,}
\label{eqn:new3.6f}
\\
r_{\pi_4}
&\le 3,
&&
\text{capacity constraint on edge $(b,d)$.}
\label{eqn:new3.6g}
\end{alignat}
\end{subequations}
%

If we choose
$\cP_1 = \{(\pi_1, 3), (\pi_2, 2) \} $
and
$\cP_2 = \{(\pi_3, 1), (\pi_4, 3) \}$,
then $\cP$ satisfies all link-capacity constraints in~\eqref{eqn:new3.5}.
If we choose $\cP_1 = \{(\pi_1, 3), (\pi_2, 4) \} $ and 
$\cP_2 = \{(\pi_3, 2), (\pi_4, 3) \}$,  then $\cP$ does not satisfy the link-capacity constraints in~\eqref{eqn:new3.5},
because the capacity constraint on edge $(u,v)$ is violated.



%
%
%

For any $k \in [K]$, we define the throughput of commodity $k$ as
%
\begin{align}
R_k (\cP) = \sum_{(\pi, r_\pi) \in \mathcal{P}_k} r_\pi.
\label{eqn:new3.7}
\end{align}
%
We call $R(\cP) = (R_1(\cP), R_2(\cP), \ldots, R_K(\cP))$ the throughput vector of
a rate allocation $\cP$.
For ease of reading, we will use $R_k$ ($R$, respectively) to denote
$R_k(\cP)$  ($R(\cP)$, respectively) when the exact meaning
can be obtained from the context without ambiguity.

In Fig.~\ref{fig:01}, if we choose
$\cP_1 = \{(\pi_1, 3), (\pi_2, 2) \} $
and
$\cP_2 = \{(\pi_3, 1), (\pi_4, 3) \}$, then $R_1=5$, $R_2=4$.
In the rest of this paper, we only concentrate on $\cP$ that 
satisfies the link-capacity constraints~\eqref{eqn:new3.5}. 
%
%




%

\subsection{Problem Definitions}
\label{subsec:problemsD}
\noindent
Let $\bF \in(\frac{1}{4},1]$ be a given fidelity threshold.
A path $\pi$ is fidelity-feasible if and only if $F(\pi) \ge \bF$.
We can translate the fidelity threshold $\bF$ into a corresponding
Werner-parameter threshold $\bW = \frac{4 \bF-1}{3}$.
Then a path $\pi$ is fidelity-feasible if and only if $W(\pi) \ge \bW$.
Let $\boldsymbol{\Delta} = -\ln \bW$.
Then the Werner-parameter constraint
\(W(\pi) \ge \bW\)
is equivalent to the additive loss constraint
$\ell(\pi)\le \boldsymbol{\Delta}$.
Since
\begin{align}
    \label{eqn:new_eq}
    F(\pi) \ge \bF \iff W(\pi) \ge \bW \iff \ell(\pi) \le \boldsymbol{\Delta}
\end{align}
We can equivalently consider the fidelity constraint
$F(\pi) \ge \bF$
or the Werner-parameter constraint
$ W(\pi) \ge \bW$
or the fidelity-loss constraint
$\ell(\pi) \le \boldsymbol{\Delta}$.

%
%
%
%
%
We study the following decision problems.

\begin{definition}[\bf RS]
\label{def:RS}
Given a quantum network $G=(V,E,c,w,q)$
, a commodity set $\mathcal{D}$ where $|\cD| = K$, a fidelity requirement $\bf F$ and a rate sum requirement $\bf T$,
the fidelity-aware \textbf{R}ate-\textbf{S}um rate allocation problem,
denoted by \textbf{RS}, asks whether there exists a rate allocation $\mathcal{P}=(\mathcal{P}_1,\mathcal{P}_2,\ldots,\mathcal{P}_K)$
such that all positive-rate paths satisfy the fidelity requirement
$W(\pi) \ge \textbf{W} = \frac{4 \textbf{F} - 1}{3}$,
$\sum_{k=1}^{K} R_k \ge \bT$,
and the link-capacity constraints~\eqref{eqn:new3.5} are satisfied.
\hfill$\Box$
\end{definition}

We use $\mathrm{RS}(G,\mathcal D, \bF,\bT)$
to denote an instance of RS, where
$G=(V,E,c,w,q)$ is the quantum network,
$\mathcal D$ is the demand set,
$\bF$ is the fidelity requirement, and
$\bT$ is the total-throughput requirement.

For the example in Fig.~1, we consider a demand set
$\mathcal D=\{(s_1,t_1),(s_2,t_2)\}$,
where $s_1=a$, $t_1=c$, $s_2=b$, and $t_2=d$.
Consider the instance
$\mathrm{RS}(G,\mathcal D,0.775,10)$.
Set
$\mathcal P_1=\{(\pi_1,3),(\pi_2,4)\}$
and
$\mathcal P_2=\{(\pi_4,3)\}$.
Then
$\mathcal P=(\mathcal P_1,\mathcal P_2)$
forms a feasible solution for
$\mathrm{RS}(G,\mathcal D,0.775,10)$.
To see this, we first note that
$\mathcal P$ satisfies all link-capacity constraints.
Moreover, all positive-rate paths satisfy
$
W(\pi)\ge 0.7,
$
which is equivalent to the fidelity requirement
$
F(\pi)\ge 0.775.
$
Since
$
R_1=7, R_2=3,
$
we have
$
R_1+R_2=10,
$
and hence the total-throughput requirement is also satisfied.

In contrast, set
$\mathcal P'_1=\{(\pi_1,3),(\pi_2,2)\}$
and
$\mathcal P'_2=\{(\pi_3,1),(\pi_4,3)\}$.
Although
$\mathcal P'=(\mathcal P'_1,\mathcal P'_2)$
satisfies the fidelity and link-capacity constraints,
it has total throughput
$
R_1+R_2=9<10.
$
Therefore, $\mathcal P'$ is not a feasible solution for
$\mathrm{RS}(G,\mathcal D,0.775,10)$.

%
%

The RS problem imposes a requirement only on total throughput and provides no per-commodity service guarantee.
We therefore introduce a second problem that adds a per-commodity minimum-rate constraint.

\begin{definition} [\bf RS-mRC]
\label{def:RS_mRC}
Given a quantum network $G=(V,E,c,w,q)$, a commodity set $\mathcal{D}$ where $|\cD| = K$, a fidelity requirement $\bF$ and a rate sum requirement $\bT$,
together with a minimum rate requirement $\bT_{\min}$,
the fidelity-aware \textbf{R}ate-\textbf{S}um subject to
\textbf{m}inimum-\textbf{R}ate-\textbf{C}onstraint rate allocation problem,
denoted by \textbf{RS-mRC},
asks whether there exists a rate allocation $\mathcal{P}=(\mathcal{P}_1,\mathcal{P}_2,\ldots,\mathcal{P}_K)$ such that all positive-rate paths satisfy
$W(\pi) \ge \bW = \frac{4 \bF - 1}{3}$,
 $\sum_{k=1}^{K} R_k \ge \bT$,
$\min\limits_{1\le k\le K} R_k \ge \bT_{\min}$,
and the link-capacity constraints~\eqref{eqn:new3.5} are satisfied.
\hfill$\Box$
\end{definition}
%

We use $\mathrm{RS\text{-}mRC}(G,\mathcal D, \bF,\bT,\bT_{\min})$
to denote an instance of RS-mRC, where
$G=(V,E,c,w,q)$ is the quantum network,
$\mathcal D$ is the demand set,
$\bF$ is the fidelity requirement,
$\bT$ is the total-throughput requirement, and
$\bT_{\min}$ is the minimum-rate requirement.

For the example in Fig.~1, we consider the same demand set
as above.
Consider the instance
$\mathrm{RS\text{-}mRC}(G,\mathcal D, \allowbreak 0.775,10,4)$.
Set
$\mathcal P_1=\{(\pi_1,3),(\pi_2,3)\}$
and
$\mathcal P_2=\{(\pi_3,1),(\pi_4,3)\}$.
Then
$\mathcal P=(\mathcal P_1,\mathcal P_2)$
forms a feasible solution for
$\mathrm{RS\text{-}mRC}(G,\mathcal D,0.775,10,4)$.
To see this, we first note that
$\mathcal P$ satisfies all link-capacity constraints.
Moreover, all positive-rate paths satisfy
$
W(\pi)\ge 0.7,
$
which is equivalent to the fidelity requirement
$
F(\pi)\ge 0.775.
$
Since
$
R_1=6, R_2=4,
$
we have
$
R_1+R_2=10
$
and
$
\min\{R_1,R_2\}=4.
$
Hence, both the total-throughput and minimum-rate requirements
are satisfied.

In contrast, set
$\mathcal P'_1=\{(\pi_1,3),(\pi_2,4)\}$
and
$\mathcal P'_2=\{(\pi_4,3)\}$.
Although
$\mathcal P'=(\mathcal P'_1,\mathcal P'_2)$
satisfies the fidelity and link-capacity constraints and has
total throughput
$
R_1+R_2=10,
$
it has
$
\min\{R_1,R_2\}=3<4.
$
Therefore, $\mathcal P'$ is not a feasible solution for
$\mathrm{RS\text{-}mRC}(G,\mathcal D,0.775,10,4)$.

%
%

Although RS-mRC provides a baseline service guarantee by imposing a minimum-rate constraint, many resource-allocation systems adopt the stronger criterion of max-min fairness~\cite{yang2010routing, joe2013multiresource}.
A system is max-min fair if no user's resource allocation can be increased
without decreasing the resource allocation of another user whose allocation
is no larger.
We formalize max-min fairness using lexicographic ordering as follows.

%

Let $\bX=(x_1,x_2,\ldots,x_K)$ and
$\bY=(y_1,y_2,\ldots,y_K)$.
We write $\bX \succ_{\rm lex} \bY$ if there exists an index $i \le K$
such that $x_i > y_i$ and $x_k=y_k$ for all $k<i$.
We write $\bX\succeq_{\mathrm{lex}}\bY$ if $\bX=\bY$ or $\bX \succ_{\rm lex} \bY$.
We write $\bX \prec_{\rm lex} \bY$ if $\bY \succ_{\rm lex} \bX$.
We write $\bX \preceq_{\rm lex} \bY$ if $\bY \succeq_{\rm lex} \bX$.
For any vector $\bX$, let $\bX^\uparrow$ denote the vector obtained by
sorting its entries in nondecreasing order.
For example, let $x=(2,4,1)$ and $y=(2,3,10)$.
Then
$x^\uparrow = (1, 2, 4)$ and $y^\uparrow = (2, 3, 10)$.
Note that
$x \succ_{\mathrm{lex}} y$, whereas
$x^\uparrow\prec_{\mathrm{lex}}y^\uparrow$.

\begin{definition}[\bf MmFair]
\label{def:MmFair}
Given a quantum network $G=(V,E,c,w,q)$, a commodity set $\mathcal{D}$ where $|\cD| = K$, a fidelity requirement $\bF$ and a nondecreasing threshold
vector
$\boldsymbol{\eta}=(\eta_1,\eta_2,\ldots,\eta_K)$,
the \textbf{M}ax-\textbf{M}in-\textbf{Fair} rate allocation problem, denoted by
\textbf{MmFair}, asks whether there exists a rate allocation $\mathcal{P}=(\mathcal{P}_1,\mathcal{P}_2,\ldots,\mathcal{P}_K)$ such that all positive-rate paths satisfy
$W(\pi)\ge \bW = \frac{4 \bF - 1}{3}$,
the link-capacity constraints~\eqref{eqn:new3.5} are satisfied, and the sorted
throughput vector satisfies
$\mathbf R^\uparrow \succeq_{\mathrm{lex}} \boldsymbol{\eta}$,
where
$\bR = (R_1, R_2, \ldots, R_K)$.
\hfill$\Box$
\end{definition}

Let $\mathbf R^{\star\uparrow}$ denote the lexicographically maximum sorted throughput vector among all rate allocations
satisfying the fidelity and link-capacity constraints.
Then an MmFair instance is feasible if and only if
$
\mathbf R^{\star\uparrow}\succeq_{\rm lex}\boldsymbol\eta.
$
Thus, MmFair asks whether the achievable max-min-fair throughput
vector meets the prescribed threshold $\boldsymbol\eta$.

We use $\mathrm{MmFair}(G,\mathcal D, \bF,\boldsymbol{\eta})$
to denote an instance of MmFair, where
$G=(V,E,c,w,q)$ is the quantum network,
$\mathcal D$ is the demand set,
$\bF$ is the fidelity requirement, and
$\boldsymbol{\eta}$ is the fairness-threshold vector.

For the example in Fig.~1, we consider the same demand set
as above.
Consider the instance
$\mathrm{MmFair}(G,\mathcal D, \allowbreak 0.775,(5,5))$.
Set
$\mathcal P_1=\{(\pi_1,3),(\pi_2,2)\}$
and
$\mathcal P_2=\{(\pi_3,2),(\pi_4,3)\}$.
Then
$\mathcal P=(\mathcal P_1,\mathcal P_2)$
forms a feasible solution for
$\mathrm{MmFair}(G,\mathcal D,0.775,(5,5))$.
To see this, we first note that
$\mathcal P$ satisfies all link-capacity constraints.
Moreover, all positive-rate paths satisfy
$
W(\pi)\ge 0.7,
$
which is equivalent to the fidelity requirement
$
F(\pi)\ge 0.775.
$
Since
$
R_1=5, R_2=5,
$
we have
$
\mathbf R^\uparrow(\mathcal P)=(5,5)
\succeq_{\mathrm{lex}}(5,5).
$
Hence, the fairness-threshold requirement is also satisfied.

In contrast, set
$\mathcal P'_1=\{(\pi_1,3),(\pi_2,3)\}$
and
$\mathcal P'_2=\{(\pi_3,1),(\pi_4,3)\}$.
Although
$\mathcal P'=(\mathcal P'_1,\mathcal P'_2)$
satisfies the fidelity and link-capacity constraints,
its commodity-throughput vector is
$
\mathbf R(\mathcal P')=(6,4),
$
and hence
$
\mathbf R^\uparrow(\mathcal P')=(4,6)
\prec_{\mathrm{lex}}(5,5).
$
Therefore, $\mathcal P'$ is not a feasible solution for
$\mathrm{MmFair}(G,\mathcal D,0.775,(5,5))$.


%
%

\subsection{Computational Complexity}
\label{subsec:Complexity}
\noindent
The authors of~\cite{chakraborty2020entanglement} presented a polynomial-time
algorithm for a special case of RS where all edges have identical Werner
parameters.
The general versions of RS, RS-mRC, and MmFair are all NP-hard, even
for the case $K=1$.
%
%
%
\begin{theorem}
\label{the:01}
The problems RS, RS-mRC, and MmFair are all NP-hard, even
for the case $K=1$.
\hfill$\Box$
\end{theorem}

\noindent
%
\textit{Proof.}
RS-mRC is equivalent to RS when $K=1$, with $\bT_{\min}=\bT$.
Similarly, MmFair is equivalent to RS when $K=1$, with $\bETA=(\bT)$.
To prove Theorem~\ref{the:01}, it suffices to prove that RS is NP-hard
when $K=1$.

We prove Theorem~\ref{the:01} by a reduction from {\bf Partition}~\cite{garey2002computers}.
An instance of Partition is given by a finite set $A$,
where each $a \in A$ is associated with a positive integer $s(a)$,
known as the \textit{size of $a$}.
It asks for the existence of a subset $A'$ of $A$ such that
$\sum_{a \in A'} s(a) = \sum_{a \in A \setminus A'} s(a)$.
This problem is known to be NP-hard~\cite{garey2002computers}.

We outline a reduction from Partition to RS in the following,
with the help of Fig.~\ref{fig:Reduction}.
Let an instance $\mathcal{I}_1$ of  Partition be given by
$A=\{a_1, a_2, \ldots, a_L\}$ and size function $s$.
We construct an instance $\mathcal{I}_2$ of RS as follows.
The set of vertices of graph $G(V, E)$ is given by
$V=\{u_0, v_1, u_1, v_2, u_2, \ldots, u_{L-1}, v_L, u_L\}$.
%
%
%
%
For each element $a_l \in A$, $E$ contains
$(u_{l-1}, u_l)$, $(u_{l-1}, v_l)$ and $(v_l, u_l)$.
We call $(u_{l-1}, u_l)$ a \textit{primary edge} (solid red in Fig.~\ref{fig:Reduction}).
It has capacity $1$ and Werner parameter $\exp (-s(a_l)) = {\ee}^{-s(a_l)}$, 
where $\ee$ is the base of the natural logarithm.
We call $(u_{l-1}, v_l)$ and $(v_l, u_l)$ a \textit{two-edge detour} (dashed blue in Fig.~\ref{fig:Reduction}).
Each has capacity $1$ and Werner parameter $1$.
%
%
%
%
%
%
Set
$q=1$,
$K = 1$,
$\bT = 2$,
$\bW = {\ee}^{- \frac{1}{2}\sum_{1 \le i \le L} s(a_i)}$,
$\bF = \frac{1 + 3 \bW}{4}$,
$s_1=u_0, t_1=u_L$.
This completes the construction of $\cI_2$.
%
%
\begin{figure}[H]
    \centering
    \includegraphics[width=0.6\linewidth]{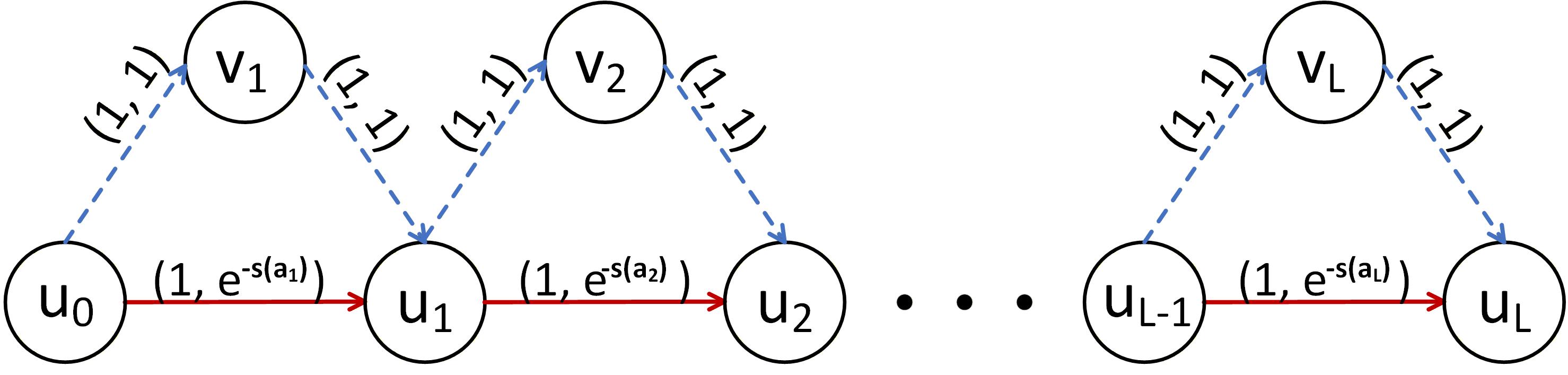}
    \vspace{-0.10in}
    \caption{Reduction from Partition to RS.}
    \label{fig:Reduction}
\end{figure}

Suppose $\mathcal{I}_1$ has a feasible solution,
{\em i.e.}, there exists $A' \subseteq A$ such that
$\sum_{a \in A'} s(a) = \sum_{a \in A \setminus A'} s(a)$.
Define two $u_0$-$u_L$ paths $\pi_1$ and $\pi_2$ as follows.
For each $l \in \{1, 2, \ldots, L\}$, $\pi_1$ uses the primary edge
$(u_{l-1}, u_{l})$ when $a_l \in A'$, and uses the detour edges
$(u_{l-1}, v_l), (v_l, u_l)$ when $a_l \not \in A'$;
$\pi_2$ uses the detour edges
$(u_{l-1}, v_l), (v_l, u_l)$ when $a_l \in A'$,
and uses the primary edge
$(u_{l-1}, u_{l})$ when $a_l \not \in A'$.
Since $\pi_1$ and $\pi_2$ do not share any edge, we can have a flow of $1$ on
each path.
Since the Werner parameter of $\pi_1$ is
${\ee}^{-\sum_{a \in A'} s(a)} = \bW$,
and the Werner parameter of $\pi_2$ is
${\ee}^{-\sum_{a \in A \setminus A'} s(a)} = \bW$,
$\{(\pi_1,1), (\pi_2,1)\}$ is a feasible solution for $\mathcal{I}_2$.

Conversely, suppose $\mathcal{I}_2$ has a feasible solution.
Let $\{\pi_1, \pi_2, \ldots, \pi_M\}$ be its positive-rate paths, 
with corresponding rates $r_1, r_2, \ldots, r_M >0$,
whose total throughput is
$\sum_{i=1}^M r_i = 2$
(the sum cannot be larger than $2$, because the max flow from $u_0$ to $u_L$ is $2$).
For each path $\pi_i$, let $B_i \subseteq A$ denote the set of elements
$a_l$ such that $\pi_i$ uses the primary edge $(u_{l-1}, u_{l})$.
%
%
The Werner parameter of $\pi_i$ is
%
$\prod_{e \in \pi_i} w(e) = \exp \left( - \sum_{a \in B_i} s(a) \right)$.
%
Since $\pi_i$ is fidelity-feasible, we have
%
\begin{align}
\exp \left( - \sum_{a \in B_i} s(a) \right)
\ge \bW =
\exp \left( - \frac{1}{2} \sum_{a \in A} s(a) \right).
\label{eqn:new3.8}
\end{align}
%
Therefore, for every $i \in \{1, 2, \ldots, M\}$, we have
%
\begin{align}
\sum_{a \in B_i} s(a) \le \frac{1}{2} \sum_{a \in A} s(a).
\label{eqn:new3.9}
\end{align}
%
Now consider any stage $l \in \{1, 2, \ldots, L\}$.
Every $u_0$-$u_L$ path must go from $u_{l-1}$ to $u_{l}$ either through
the primary edge $(u_{l-1}, u_l)$
or through the two-edge detour $(u_{l-1}, v_l)$, $(v_l, u_l)$.
The total capacity of these two alternatives is $2$,
and the total feasible throughput is $2$.
Hence, in every stage $l$, both alternatives must be saturated.
In particular, the total flow using the primary edge $(u_{l-1}, u_l)$ is exactly $1$.
Thus, for every $l \in \{1, 2, \ldots, L\}$,
%
$\sum_{i: a_l \in B_i} r_i = 1$.
%
It follows that
%
\begin{align}
\sum_{i=1}^M r_i \sum_{a_l \in B_i} s(a_l)
=
%
%
\sum_{l=1}^L s(a_l) \sum_{i: a_l \in B_i} r_i
=
%
%
\sum_{l=1}^L s(a_l).
\label{eqn:new3.10}
\end{align}
%
Since $\sum_{i=1}^M r_i = 2$, the rate-weighted average of the quantities
$\sum_{a \in B_i} s(a)$ is
%
\begin{align}
\frac{\sum_{i=1}^M r_i \sum_{a_l \in B_i} s(a_l)}{\sum_{i=1}^M r_i}
= \frac{1}{2} \sum_{l=1}^L s(a_l)
= \frac{1}{2} \sum_{a \in A} s(a).
%
%
\label{eqn:new3.11}
\end{align}
%
Equality (\ref{eqn:new3.11}) and inequality (\ref{eqn:new3.9}) imply that
%
\begin{align}
\sum_{a \in B_i} s(a) = \frac{1}{2} \sum_{a \in A} s(a),~
i=1, 2, \ldots, M.
%
%
\label{eqn:new3.12}
\end{align}
%
Therefore $A' = B_1$ is a feasible solution to $\cI_1$.
This proves that RS is NP-hard even for $K=1$.
\hfill$\Box$

\section{Optimization problems}
\label{sec:opt_def}
\noindent
In this section, we define the optimization versions of the three
decision problems.
Unlike the decision problems, where a fidelity threshold is given as part of the input, the optimization problems seek to maximize the minimum end-to-end fidelity among all positive-rate paths while meeting a prescribed throughput or fairness requirement. 
Equivalently, under the fidelity-loss representation, they minimize the maximum path fidelity loss.
We develop fully polynomial-time approximation schemes for these
optimization problems in the subsequent sections.

Let $\cP_k$ be a rate allocation for demand $(s_k, t_k)$ where $k \in [K]$.
The Werner parameter, the fidelity, and the fidelity loss of $\cP_k$ are respectively:
%
\begin{align}
    W(\cP_k) = \min_{\substack{(\pi,r_\pi)\in\mathcal P_k:\\r_\pi>0}} W(\pi),
    \label{eqn:new4.1}
\end{align}

\begin{align}
F(\mathcal P_k)
&=
\frac{1+3W (\mathcal P_k)}{4}.
\label{eqn:new4.2}
\end{align}

\begin{align}
\label{eqn:new4.3}
\Delta(\cP_k) = \max_{\substack{(\pi,r_\pi)\in\mathcal P_k:\\r_\pi>0}}
\ell(\pi).
\end{align}
%
Note that
$W(\cP_k)$ is the minimum over the Werner parameter of $s_k$-$t_k$
path $\pi$ such that $(\pi, r_\pi) \in \cP_k$ where $r_\pi > 0$;
$F(\cP_k)$ is the minimum over the fidelity of $s_k$-$t_k$
path $\pi$ such that $(\pi, r_\pi) \in \cP_k$ where $r_\pi > 0$;
$\Delta(\cP_k)$ is the maximum over the fidelity loss of $s_k$-$t_k$
path $\pi$ such that $(\pi, r_\pi) \in \cP_k$ where $r_\pi > 0$.

In the extreme case where $r_\pi = 0$ for every $(\pi, r_\pi) \in \cP_k$,
we define
$W(\cP_k) = 1$,
$F(\mathcal P_k) = 1$,
and
$\Delta(\cP_k) = 0$.
This convention is slightly different from the mathematical convention that
the maximum over an empty set is $- \infty$
and
the minimum over an empty set is $\infty$.
We adopt these values as a bookkeeping convention for an inactive commodity.
They lie within the physical ranges of the corresponding quantities and make an inactive commodity neutral in the aggregate definitions below.

For a rate allocation
$\mathcal P=(\mathcal P_1,\mathcal P_2,\ldots,\mathcal P_K),$
the minimum Werner parameter, the minimum fidelity, and the maximum fidelity loss of $\cP$ are respectively:
%
\begin{align}
    W(\cP) = \min_{k\in[K]} W (\cP_k),
    \label{eqn:new4.4}
\end{align}

\begin{align}
F(\mathcal P)
&=
\frac{1+3W (\mathcal P)}{4}.
\label{eqn:new4.5}
\end{align}
\begin{align}
\label{eqn:new4.6}
\Delta(\cP) = \max_{k\in[K]} \Delta(\cP_k).
\end{align}
%
%
In the extreme case where $r_\pi = 0$ for every $(\pi, r_\pi) \in \cP_k$ and every $k \in [K]$,
we define
$W(\cP) = 1$,
$F(\mathcal P) = 1$,
and
$\Delta(\cP) = 0$,
as a convention.
%

%
Note that only positive-rate paths are included in the above definitions, since a
zero-rate path can be removed from the allocation without affecting its
throughput or capacity consumption.
Note that maximizing $F(\mathcal P)$ is equivalent to 
maximizing $W(\mathcal P) = {\ee}^{-\Delta(\cP)}$,
which is equivalent to
minimizing $\Delta(\cP)$.


\subsection{Optimization Versions}
\label{subsec:4.1}
\noindent
We define the optimization versions of RS, RS-mRC and MmFair as follows.
\begin{definition}[$\tRSo$]
\label{def:ors}
Given a quantum network $G=(V,E,c,w,q)$, a commodity set $\mathcal{D}$ where $|\cD| = K$, and a total-throughput requirement $\bT$, the 
optimization version of
RS, denoted by
$
\textbf{$\tRSo$},
$
asks for a rate allocation
$
\mathcal P=(\mathcal P_1,\ldots,\mathcal P_K)
$
that minimizes
$
\Delta(\mathcal P)
$
subject to
$\sum_{k=1}^{K}R_k\ge\bT$ and link-capacity constraints.
\hfill$\Box$
\end{definition}

We use $\toRS(G,\mathcal D,\bT)$ to denote an instance of $\toRS$, where $G=(V, E, c, w, q)$ is the quantum network, $\mathcal{D}$ is the demand set, and  $\bT$ is the total-throughput requirement.
We use $\Delta^{\toRS}_{\mathrm{opt}}(G, \mathcal{D}, \bT)$
to denote the optimal value of
$\tRSo(G,\mathcal D,\bT)$.
For ease of reading, we may simply use 
$\Delta^{\toRS}_{\mathrm{opt}}$
to denote
$\Delta^{\toRS}_{\mathrm{opt}}(G, \mathcal{D}, \bT)$,
if the exact meaning can be obtained from the context without ambiguity.

For the example in Fig.~\ref{fig:01}, we consider a demand set
$\cD = \{(s_1, t_1), (s_2, t_2)\}$ where
$s_1 = a$, $t_1 = c$, $s_2 = b$, $t_2 = d$.
Consider the instance $\tRSo(G,\mathcal D, 4)$.
Set $\cP_1=\{(\pi_2,4)\}$ and $\cP_2= \emptyset$.
Then $\cP = (\cP_1, \cP_2)$ forms an optimal solution for $\toRS(G, \cD, 4)$,
with optimal fidelity loss
$\Delta^{\toRS}_{\mathrm{opt}}(G, \mathcal{D}, 4)=-\ln 0.95$.
To see why $\cP = (\cP_1, \cP_2)$ is an optimal solution for $\toRS(G, \cD, 4)$,
we first note that $\cP$ satisfies the link-capacity constraints in (\ref{eqn:new3.5}).
Since $\sum_{k=1}^2 R_k(\cP) = 4$, $\cP$ satisfies the throughput constraint.
Since the minimum fidelity-loss of any $s_1$-$t_1$ path is $- \ln 0.95$
and the minimum fidelity-loss of any $s_2$-$t_2$ path is $- \ln 0.9 > - \ln 0.95$,
$\cP$ is an optimal solution for instance $\tRSo(G,\mathcal D, 4)$.

%
%

Consider the instance $\tRSo(G,\mathcal D, 6)$.
Set $\cP'_1=\{(\pi_2,4)\}$ and $\cP'_2=\{(\pi_4,2)\}$.
Then $\cP' = (\cP'_1, \cP'_2)$ forms an optimal solution for $\tRSo(G,\mathcal D, 6)$ with  
$\Delta^{\toRS}_{\mathrm{opt}}(G, \mathcal{D}, 6)=-\ln 0.7$.
%
%

\begin{definition}[$\tRSmRCo$]
\label{def:ors_mrc}
Given a quantum network $G=(V,E,c,w,q)$, a commodity set $\mathcal{D}$ where $|\cD| = K$, a total-throughput requirement $\bT$ and a minimum-rate requirement
$\bT_{\min}$, the optimization version of RS-mRC, denoted by
$
\textbf{$\tRSmRCo$}
$,
asks for a rate allocation
$
\mathcal P=(\mathcal P_1,\mathcal P_2,\ldots,\mathcal P_K)
$
that minimizes
$
\Delta(\mathcal P)
$
subject to
$
\sum_{k=1}^{K}R_k\ge\bT,
$
$\min\limits_{1\le k\le K} R_k \ge \bT_{\min}$
and all link-capacity constraints.
\hfill$\Box$
\end{definition}

We use $\tRSmRCo(G,\mathcal D,\bT,\bT_{\min})$ to denote
an instance of $\tRSmRCo$, where $G=(V, E, c, w, q)$ is the quantum network, $\mathcal{D}$ is the demand set, $\bT$ is the total-throughput requirement, 
and $\bT_{\min}$ is the minimum rate requirement.
We use 
$
\Delta_{\mathrm{opt}}^{\mathrm{mRC}}(G, \cD, \bT, \bT_{\min})
$
to denote the optimal value of
$
\tRSmRCo \allowbreak(G,\mathcal D,\bT,\bT_{\min}).
$
For ease of reading, we may simply use 
$\Delta^{\mathrm{mRC}}_{\mathrm{opt}}$
to denote
$\Delta^{\mathrm{mRC}}_{\mathrm{opt}}(G, \mathcal{D}, \bT, \bT_{\min})$,
if the exact meaning can be obtained from the context without ambiguity.

For the example in Fig.~\ref{fig:01}, we consider the same demand set
as above.
Consider the instance $\tRSmRCo(G,\mathcal D, 4, 2)$.
Set 
$\cP_1=\{(\pi_2,2)\}$ and $\cP_2=\{(\pi_3,2)\}$.
Then $\cP=(\cP_1, \cP_2)$ forms an optimal solution for
$\tRSmRCo(G, \mathcal{D}, 4, 2)$ with
$\Delta^{\mathrm{mRC}}_{\mathrm{opt}}(G, \mathcal{D}, 4, 2)=-\ln 0.9$.
We can verify that $\cP$ satisfies the link-capacity constraints in (\ref{eqn:new3.5}).
Since $R_1(\cP) = 2$ and $R_2(\cP) = 2$,
we have
$\sum_{k=1}^2 R_k(\cP) \ge 4$
and
$\min_{1 \le k \le 2} R_k(\cP) \ge 2$.
Therefore $\cP$ satisfies both the total-throughput constraint ($\bT = 4$)
and the minimum rate constraint ($\bT_{\min} = 2$).
Moreover, commodity~2 must receive rate at least $2$, and its
highest-Werner-parameter path is $\pi_3$, with $W(\pi_3)=0.9$.
Hence, no feasible allocation can have maximum path loss smaller than
$-\ln 0.9$.
Therefore,
$\Delta_{\mathrm{opt}}^{\mathrm{mRC}}=-\ln 0.9$.

Consider the instance  $\tRSmRCo(G, \mathcal{D}, 6, 3)$.
Set 
$\cP'_1=\{(\pi_1,1), (\pi_2,2)\}$ and $\cP'_2=\{(\pi_3,2), (\pi_4,1)\}$.
Then $\cP' = (\cP'_1, \cP'_2)$ forms an optimal solution for
$\tRSmRCo(G, \mathcal{D}, 6, 3)$ 
with  $
\Delta^{\mathrm{mRC}}_{\mathrm{opt}}(G, \mathcal{D}, 6, 3)=-\ln 0.7$.

\begin{definition}[$\tMmFairo$]
\label{def:ommfair}

Given a quantum network $G=(V,E,c,w,q)$, a commodity set $\mathcal{D}$ where $|\cD| = K$, and a fairness-threshold vector
$\boldsymbol{\eta} = (\eta_1,\eta_2,\ldots,\eta_K), $
the optimization version of MmFair, denoted by
$
\textbf{$\tMmFairo$},
$
asks for a rate allocation
$
\mathcal P=(\mathcal P_1,\mathcal P_2,\ldots,\mathcal P_K)
$
that minimizes $\Delta(\mathcal P)$  subject to
$
\mathbf R^\uparrow(\mathcal P)
\succeq_{\mathrm{lex}}
\boldsymbol{\eta}
$,
and all link-capacity constraints.
\hfill$\Box$
\end{definition}

We use $\tMmFairo(G,\mathcal D,\boldsymbol{\eta})$
to denote an instance of $\tMmFairo$, where
$G=(V,E,c,w,q)$ is the quantum network,
$\mathcal D$ is the demand set, and
$\boldsymbol{\eta}$ is the fairness-threshold vector.
We use
$\Delta_{\mathrm{opt}}^{\mathrm{fair}}
(G,\mathcal D,\boldsymbol{\eta})$
to denote the optimal value of
$\tMmFairo(G,\mathcal D,\boldsymbol{\eta})$.
For ease of reading, we may simply use
$\Delta_{\mathrm{opt}}^{\mathrm{fair}}$
to denote
$\Delta_{\mathrm{opt}}^{\mathrm{fair}}
(G,\mathcal D,\boldsymbol{\eta})$,
if the exact meaning can be obtained from the context without ambiguity.

For the example in Fig.~\ref{fig:01}, we consider the same demand set
as above.
Consider the instance
$\tMmFairo(G,\mathcal D,(2,2))$.
Set
$\cP_1=\{(\pi_2,2)\}$ and
$\cP_2=\{(\pi_3,2)\}$.
Then $\cP=(\cP_1,\cP_2)$ forms an optimal solution for
$\tMmFairo(G,\mathcal D,(2,2))$,
with optimal fidelity loss
$
\Delta_{\mathrm{opt}}^{\mathrm{fair}}
(G,\mathcal D,(2,2))
=
-\ln 0.9.
$
To see why $\cP=(\cP_1,\cP_2)$ is an optimal solution for
$\tMmFairo(G,\mathcal D,(2,2))$,
we first note that $\cP$ satisfies the link-capacity constraints
in~(\ref{eqn:new3.5}).
Since
$
\mathbf R^\uparrow(\cP)=(2,2) \succeq_{\mathrm{lex}} \boldsymbol{\eta},
$
$\cP$ satisfies the fairness-threshold constraint.
Moreover, satisfying
$
\mathbf R^\uparrow(\cP)\succeq_{\mathrm{lex}}(2,2)
$
requires both commodities to receive rate at least $2$.
The minimum fidelity loss of any $s_2$-$t_2$ path is
$-\ln 0.9$.
Hence, no feasible allocation can have maximum fidelity loss
smaller than $-\ln 0.9$.
Therefore, $\cP$ is an optimal solution for
$\tMmFairo(G,\mathcal D,(2,2))$.

Consider the instance
$\tMmFairo(G,\mathcal D,(3,3))$.
Set
$\cP'_1=\{(\pi_1,1),(\pi_2,2)\}$
and
$\cP'_2=\{(\pi_3,2),(\pi_4,1)\}$.
Then
$\cP'=(\cP'_1,\cP'_2)$
forms an optimal solution for
$\tMmFairo(G,\mathcal D,(3,3))$
with
$
\Delta_{\mathrm{opt}}^{\mathrm{fair}}
(G,\mathcal D,(3,3))
=
-\ln 0.7.
$

Note that $\tMmFairo$ does not lexicographically maximize the throughput
vector as its objective.
The fairness threshold $\boldsymbol{\eta}$ is prescribed, and the
objective is to minimize the maximum fidelity loss while meeting this
threshold.
The lexicographically maximum vector is computed later only to test
feasibility under a fixed loss threshold.

\subsection{Relationship to the Decision Problems}
\noindent
If no rate allocation satisfies the corresponding throughput or
fairness requirements together with the link-capacity constraints,
the optimization instance is infeasible.
\begin{proposition}
\label{prop:threshold_monotonicity}
For each of RS, RS-mRC, and MmFair, feasibility is monotone in the
fidelity-loss threshold.
Moreover, if the corresponding optimization problem is feasible, its
optimal value is the smallest loss threshold under which the decision
problem is feasible.
In particular, for any $\boldsymbol{\Delta}\ge 0$, the corresponding decision instances
are feasible if and only if
$
\Delta^{\toRS}_{\mathrm{opt}}\le \boldsymbol{\Delta},
\Delta_{\mathrm{opt}}^{\mathrm{mRC}}\le \boldsymbol{\Delta},
\Delta_{\mathrm{opt}}^{\mathrm{fair}}\le \boldsymbol{\Delta},
$
respectively.
\hfill$\Box$
\end{proposition}
\begin{proof}
Consider first RS.
A rate allocation $\cP$ is feasible under loss threshold $\boldsymbol{\Delta}$ if and
only if it satisfies the throughput and link-capacity constraints and
$
\Delta(\mathcal P)\le \boldsymbol{\Delta}.
$
By the definition of $\Delta^{\toRS}_{\mathrm{opt}}$, such an allocation exists
if and only if
$
\Delta^{\toRS}_{\mathrm{opt}}\le \boldsymbol{\Delta}.
$
The same argument applies to RS-mRC and MmFair.

Furthermore, if a rate allocation $\cP$ is feasible under threshold $\boldsymbol{\Delta}$, then
$\cP$ remains feasible under every threshold
$\boldsymbol{\Delta}'\ge \boldsymbol{\Delta}$.
Hence, feasibility is monotone in the fidelity-loss threshold.
\end{proof}

\subsection{Approximation Parameters}
\noindent
%
In this paper, we will design fully polynomial-time approximation schemes (FPTAS)
for the three optimization problems defined in Section~\ref{sec:opt_def}-A.
For any feasible instance of the three optimization problems, let
$F_{\mathrm{opt}}$,
$W_{\mathrm{opt}}$,
and
$\Delta_{\mathrm{opt}}$
denote the fidelity, Werner parameter, and fidelity loss of the optimal rate allocation,
respectively.
Then we must have
%
\begin{align}
\label{eqn:new4.7}
W_{\mathrm{opt}} = \frac{4 F_{\mathrm{opt}} - 1}{3}, \text{~and~} \Delta_{\mathrm{opt}} = - \ln W_{\mathrm{opt}}
\end{align}
%
For any feasible rate allocation $\cP$, we have
%
\begin{align}
\label{eqn:new4.8}
W (\cP) = \frac{4 F(\cP) - 1}{3}, \text{~and~} \Delta(\cP) = - \ln W(\cP)
\end{align}
%

\begin{lemma}
\label{lem:new01}
For any $\delta \in (0, 1)$, set $\epsilon = \frac{\ee}{3}\delta$,
where $\ee$ is the base of natural logarithm. 
Then $\epsilon \in (0, 1)$.
Furthermore, let $\cP$ be any feasible rate allocation such that
%
\begin{align}
\label{eqn:new4.9}
\Delta(\cP) \le (1 + \epsilon) \Delta_{\mathrm{opt}}.
\end{align}
%
Then
%
\begin{align}
\label{eqn:new4.10}
F(\cP) \ge (1 - \delta) F_{\mathrm{opt}}.
\end{align}
\hfill$\Box$
\end{lemma}
%
%
\noindent
%
%
For ease of reading, we leave the proof of this lemma in the appendix.
By the standard definition~\cite{kleinberg2006algorithm, garey2002computers},
an FPTAS for $\tRSo$ ($\tRSmRCo$, $\tMmFairo$) can compute a rate allocation
$\cP$ such that
$F(\cP) \ge (1-\delta) F_{\mathrm{opt}}$
while satisfying the link-capacity
constraints and the throughput or fairness requirement, for any  $\delta \in (0,1)$,
with running time bounded by a polynomial in the instance size and $\frac{1}{\delta}$.
When $\epsilon = \frac{\ee}{3} \delta$, $\frac{1}{\epsilon} = \frac{3}{\ee} \times \frac{1}{\delta}$.
Hence a running time bounded by a polynomial in the instance and $\frac{1}{\epsilon}$
is also bounded by a polynomial in the instance and $\frac{1}{\delta}$.
By Lemma~\ref{lem:new01}, it suffices to design an FPTAS in the fidelity-loss domain that can
compute a rate allocation $\cP$ such that
$\Delta (\cP) \le (1 + \epsilon) \Delta_{\mathrm{opt}}$
while satisfying the link-capacity
constraints and the throughput or fairness requirement, for any $\epsilon \in (0,1]$,
with running time bounded by a polynomial in the instance size and $\frac{1}{\epsilon}$.
%
%


%
\section{Integer Version of RS}
\label{sec:integer_loss_exact}
\noindent
We consider a special case of RS in which
the fidelity loss of an edge takes integer values.
We develop a pseudo-polynomial-time algorithm to solve iRS.
This algorithm is a building block for the FPTAS for the general problems.

\subsection{Integer RS}
\noindent
For each physical edge \(e\in E\), recall that its fidelity loss is
$\ell(e) = -\ln w(e) .$
For computational purposes, we assume that each edge loss $\ell(e)$ is a rational
number represented in binary.
Since $w(e)=e^{-\ell(e)}$, in this section we equivalently represent the
network as $G=(V,E,c,\ell,q)$ and consider a special case of RS, denoted by $\textbf{iRS}$, in which
the fidelity loss of an edge takes integer values:
%
%
$\ell(e) \in \mathbb{Z}_{\ge 0}, \forall e \in E$.
%
%
Let $L\in \mathbb{Z}_{\ge 0}$ be an integer fidelity-loss budget.
We denote by $\textbf{iRS}(G, \cD, L, \bT)$  an instance of iRS with throughput
requirement $\bT$ and fidelity-loss constraint $L$.



\subsection{Hop-Fidelity Expanded Graph }
\noindent
%
%
%
The authors of~\cite{chakraborty2020entanglement} designed a polynomial-time
multicommodity-flow-based algorithm to solve a special case of RS where the
Werner parameters on all edges are identical.
Their algorithm used the concept of layered graph, which appears in Ford and
Fulkerson~\cite[Ch.~III, \S9, pp.~142-151]{ford1962flows}
in the study of maximal dynamic flows
and was later used in~\cite{misra2009polynomial, xue2008polynomial} for delay-constrained multipath routing.
To capture the effect of both hop-count (for entanglement swapping) and fidelity-loss,
we use two state dimensions.
The fidelity-loss coordinate is needed to enforce the path fidelity constraint,
while the hop-count coordinate is needed because a path with $h$ physical edges
requires $h-1$ swapping operations and hence has end-to-end success probability
$q^{h-1}$.
Accordingly, we generalize the layered-graph idea to a
\textit{hop-fidelity expanded graph}
$G_L^{\cD}=(V_L^{\cD},E_L^{\cD})$
for given $\cD$ and $L$, where
\begin{align}
\nonumber
V_L^{\cD} = \left\{
u^{[h,d]}: u \in V, 0 \le h < n, 0 \le d \le L, h,d \in\mathbb Z
\right\}.
\end{align}

Here, $u^{[h,d]}$ denotes the state of being at node $u$ after traversing
$h$ physical hops from the source and accumulating fidelity-loss cost $d$.
For each \textit{physical edge} $e = (u, v) \in E$, $E_L^{\cD}$ contains an
\textit{expanded edge} from
$u^{[h,d]}$ to $v^{[h+1, d+\ell(e)]}$ provided that $h+1 \le n-1$ and $d+\ell(e) \le L$:
%
%
\begin{equation}
\nonumber
\begin{aligned}
E_L^{\cD} = \Bigl\{
&\bigl(u^{[h,d]},v^{[h+1,d+\ell(e)]}\bigr):
e=(u,v)\in E, 
&0\le h < n-1,\quad
0\le d\le L-\ell(e)
\Bigr\}.
\end{aligned}
\end{equation}

Since each edge in $E^{\cD}_L$ is in the form $\bigl(u^{[h,d]},v^{[h+1,d+\ell(e)]}\bigr)$,
the hop-fidelity expanded graph $G^{\cD}_L$ is a directed acyclic graph (DAG),
due to the hop-count increase from $h$ to $h+1$.
The acyclic nature makes the hop-fidelity expanded graph $G^{\cD}_L$ a useful
tool in solving $\tiRS$.
%
%
More importantly, a physical $s_k$-$t_k$ path in $G$ with $h$ hops and integer
fidelity loss $d \le L$ corresponds to a path from $s_k^{[0,0]}$ to $t_k^{[h,d]}$
in the expanded graph.
Conversely, every path from $s_k^{[0,0]}$ to $t_k^{[h,d]}$ in the expanded graph
projects to an $s_k$-$t_k$ walk in $G$ with $h$ hops and fidelity loss $d$.

We use a simple example to illustrate the construction of $G^{\cD}_L$.

\begin{figure}[htbp]
    \centering
    \includegraphics[width=0.5\linewidth,height=1.5in]{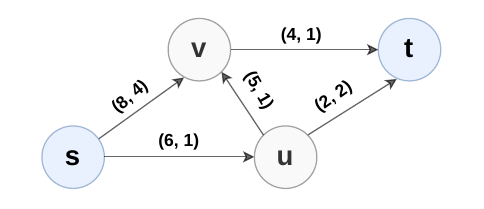}
    \caption{Physical graph $G=(V, E, c, \ell, q)$.
    The label on edge $e$
    denotes $(c(e), \ell (e))$.}
    \label{fig:original1}
\end{figure}

Fig.~\ref{fig:original1} shows the graph $G$.
There are $4$ vertices and $5$ edges in the graph.
The label on each edge $e$ shows $(c(e), \ell(e))$.
Assume $\cD = \{(s_1, t_1)\} = \{(s, t)\}$ and $L=3$.
The hop-fidelity expanded graph $G^{\cD}_L$ is shown in Fig.~\ref{fig:layer}.
The physical path $s \to u \to t$ in $G$ corresponds to the expanded path
$s^{[0,0]}$$\to$$u^{[1,1]}$$\to$$t^{[2,3]}$ in $G_L^{\cD}$.
The physical path $s \to u \to v \to t$ in $G$ corresponds to the expanded path
$s^{[0,0]}$$\to$$u^{[1,1]}$$\to$$v^{[2,2]}$$\to$$t^{[3,3]}$ in $G_L^{\cD}$.
Conversely, the expanded path
$s^{[0,0]}$$\to$$u^{[1,1]}$$\to$$t^{[2,3]}$ in $G_L^{\cD}$
projects to the physical path 
$s \to u \to t$ in $G$,
and the expanded path $s^{[0,0]}$$\to$$u^{[1,1]}$$\to$$v^{[2,2]}$$\to$$t^{[3,3]}$ in $G_L^{\cD}$
projects to the physical path $s \to u \to v \to t$ in $G$.
%

\begin{figure*}[t]
    \centering
    \includegraphics[width=1.0\linewidth, height=2.3in]{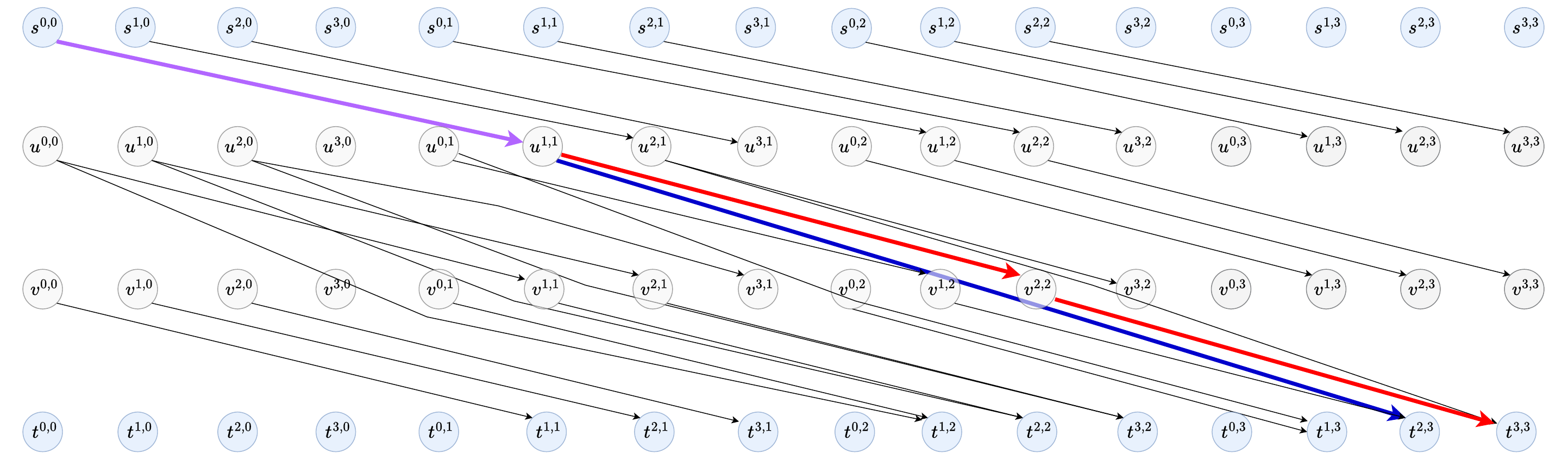}
    \caption{Expanded graph $G_L^{\cD}$ corresponding to the graph $G$ in
    Fig.~\ref{fig:original1}
    with $\cD = \{(s_1, t_1)\} = \{(s, t)\}$ and $L=3$.
    %
    %
    %
    %
    The 3-hop path $s^{[0,0]}$$\to$$u^{[1,1]}$$\to$$v^{[2,2]}$$\to$$t^{[3,3]}$ in $G_L^{\cD}$
    (using purple and red edges)
    corresponds to the path
    $s$$\rightarrow$$u$$\rightarrow$$v$$\rightarrow$$t$ in $G$.
    The 2-hop path $s^{[0,0]}$$\to$$u^{[1,1]}$$\to$$t^{[2,3]}$ in $G_L^{\cD}$
    (using purple and blue edges)
    corresponds to the path
    $s$$\rightarrow$$u$$\rightarrow$$t$ in $G$.
    %
    %
    %
    %
    }
\label{fig:layer}
\end{figure*}

\subsection{An Exact Algorithm for $\tRSi$}
\noindent
%
%
%
%
Our algorithm for iRS is based on a Linear Programming (LP) formulation.
For an instance $\tRSi(G, \cD, L, \bT)$, we construct the hop-fidelity expanded
graph $G_L^{\cD}$, then solve an LP problem based on flows $f_k$ from
$s_k^{[0, 0]}$ to terminal nodes
$\Omega_k^{L} = \{t_k^{[h,d]}: 1 \le h < n,\; 0 \le d \le L\}$,
$k \in [K]$.
For each commodity $k$ and expanded edge $(\mu, \nu)$, let $f_k(\mu, \nu)$ denote
the elementary-pair generation rate allocated to commodity $k$ on that edge.
The resulting LP is formally presented in (\ref{lp:5.1}).
For clarity, we use $u$ and $v$ to denote physical nodes in $G$, and use 
$\mu$, $\nu$, $\xi$, and  $v^{[h,d]}$ for nodes in the expanded graph.
%

%
\begin{subequations}
\label{lp:5.1}
\begin{align}
\mathrm{LP}_{\mathrm{RS}}(G, \cD, L):
\quad
\text{max}~
&
\Phi_L =
\sum_{k=1}^{K}
\sum_{h=1}^{n-1}
\sum_{d=0}^{L}
q^{h-1}
\sum_{\mathclap{(\mu,t_k^{[h,d]})\in E_L^{\cD}}}
f_k(\mu,t_k^{[h,d]})
\label{lp:5.1a}
\\[1mm]
\text{s.t.}~
&
\sum_{(\mu,\nu)\in E_L^{\cD}}\! f_k(\mu,\nu)
=
\sum_{(\nu,\xi)\in E_L^{\cD}}\! f_k(\nu,\xi),
\quad \forall \nu \in V_L^{\cD} \setminus (
\{s_k^{[0,0]}\}\cup \Omega_k^{L}
),
\forall k \in [K],
\label{lp:5.1b}
\\[2mm]
&
\sum_{k=1}^{K}
\sum_{h=0}^{n-2}
\sum_{d=0}^{L-\ell(e)}
f_k\!\left(
u^{[h,d]},
v^{[h+1,d+\ell(e)]}
\right)
\le c(e),
\quad \forall e = (u, v) \in E: \ell(e) \le L,
\label{lp:5.1c}
\\
&
f_k(\mu, \nu) = 0,~~ \forall (\mu,\nu)\in E_L^{\mathcal D}, \quad \forall \mu \in \Omega_k^L, \quad \forall k \in [K],
\label{lp:5.1d}
\\
&
f_k(\mu,\nu)\ge 0,~~
\quad \forall k,
\quad \forall(\mu,\nu)\in E_L^{\cD}.
\label{lp:5.1e}
\\ \nonumber
\end{align}
\end{subequations}

%

The objective~\eqref{lp:5.1a} maximizes the total delivered
throughput over all commodities.
An expanded terminal state $t_k^{[h,d]}$ represents reaching destination
$t_k$ after traversing $h$ physical edges and accumulating fidelity loss
$d$.
If $x$ units of elementary-pair generation flow reach this state, the
corresponding path requires $h-1$ entanglement-swapping operations.
Since each swapping operation succeeds independently with probability
$q$, only
$
q^{h-1}x
$
units of end-to-end entanglement are successfully delivered.
Therefore, the incoming flow collected at $t_k^{[h,d]}$ is weighted by
$q^{h-1}$.
Summing over all terminal states, all feasible loss values, and all
commodities gives the total delivered throughput in
\eqref{lp:5.1a}.

For each $k \in [K]$, define
%
\begin{align}
\label{eqn:5.2}
R_k(f) \!\! = \!\!
\sum_{h=1}^{n-1}
\sum_{d=0}^{L} q^{h-1}
\sum_{\mathclap{(\mu,t_k^{[h,d]})\in E_L^{\cD}}}
f_k(\mu,t_k^{[h,d]}).
\end{align}

%
$R_k(f)$ is the delivered throughput of commodity $k$.
Consequently,
\begin{align}
\Phi_L =
\sum_{k=1}^{K}R_k(f).
\label{eq:5.3}
\end{align}

Constraint~\eqref{lp:5.1b} enforces flow conservation at
all non-source and nonterminal states.
Constraint~\eqref{lp:5.1c} imposes link capacities.
The formulation also applies to undirected graphs by replacing each undirected edge with two oppositely directed edges sharing the same edge capacity, and the link-capacity constraint~\eqref{lp:5.1c} can be replaced by a joint constraint requiring the total flow in both
directions not to exceed the capacity.
Note that we impose the capacity constraint only for edges $e\in E$ with
$\ell(e)\le L$, since edges with $\ell(e)>L$ have no corresponding
expanded edges in $G_L^{\mathcal D}$.
Constraint~\eqref{lp:5.1d} prevents flow from leaving the terminal states.

For the example in Fig.~\ref{fig:original1}, consider the instance
$\tiRS(G, \cD, L, \bT)$
with 
$\mathcal D=\{(s,t)\}$, $L=3$, $\bT=2$ and $q=0.5$.
We construct the hop-fidelity expanded graph $G^{\cD}_L$ and solve the
corresponding $\mathrm{LP}_{\mathrm{RS}}(G, \cD, L)$.
An optimal solution of $\mathrm{LP}_{\mathrm{RS}}(G, \cD, L)$ is
$f_1(s^{[0,0]},u^{[1,1]})=6,
f_1(u^{[1,1]},t^{[2,3]})=2,
f_1(u^{[1,1]},v^{[2,2]})=4,
f_1(v^{[2,2]},\allowbreak t^{[3,3]})  =4,
$
with all other flow variables equal to zero.
The corresponding objective value is
$\Phi_3^\star = 0.5\times2+0.5^2\times4 = 2$.

%
%

Following the chain-decomposition procedure of Ford and Fulkerson~\cite{ford1962flows},
we can decompose $f$ into expanded paths.
%
%
%
%
Starting from a terminal node with positive incoming flow, e.g. $t^{[2,3]}$, and tracing
positive-flow incoming edges backward to
$s^{[0,0]}$ gives the expanded path
$
\widehat{\pi}_1:
s^{[0,0]}\to u^{[1,1]}\to t^{[2,3]}.
$
The bottleneck flow value of $\widehat{\pi}_1$ is $b(\widehat{\pi}_1)=2$.
We then set $f_1(e) \leftarrow f_1(e)-b(\widehat{\pi}_1)$ for
every edge $e$ on $\widehat{\pi}_1$ to obtain the residual flow.
We continue this process until no terminal node has a positive incoming flow.
The process then yields
$
\widehat{\pi}_2:
s^{[0,0]}\to u^{[1,1]}\to v^{[2,2]}\to t^{[3,3]},
$
with 
$
b(\widehat{\pi}_2)=4
$ in addition to $\widehat{\pi}_1$ with $b(\widehat{\pi}_1)=2$.

Each extracted expanded path is then projected onto the physical graph
by removing the hop and fidelity-loss indices of its vertices.
For example,
$
\widehat{\pi}_1:
s^{[0,0]}\to u^{[1,1]}\to t^{[2,3]}
$
projects to
$
\pi_1:s\to u\to t,
$
while
$
\widehat{\pi}_2:
s^{[0,0]}\to u^{[1,1]}\to v^{[2,2]}\to t^{[3,3]}
$
projects to
$
\pi_2:s\to u\to v\to t.
$
The terminal states show that $\pi_1$ and $\pi_2$ have respectively
$2$ and $3$ hops, both with fidelity loss $3$.
Thus, their delivered rates are
$
r_{\pi_1}=q^{2-1}b(\widehat{\pi}_1)=0.5\times2=1
$
and
$
r_{\pi_2}=q^{3-1}b(\widehat{\pi}_2)=0.5^2\times4=1.
$
Hence, the extracted rate allocation is
$
\mathcal P_1=\{(\pi_1,1),(\pi_2,1)\},
$
with total delivered throughput $R_1=2$.
We formalize this path extraction process in the following.

\noindent
\textbf{Path extraction.}
In general, we extract a rate allocation from an optimal solution $f$ of
$\mathrm{LP}_{\mathrm{RS}}(G, \cD, L)$, following the chain-decomposition 
procedure of Ford and Fulkerson~\cite[Ch.~I, \S2, pp.~4-9]{ford1962flows}.
For each commodity $k \in [K]$ we initialize $\cP_k$ to $\emptyset$.
Then, select any terminal state $t_k^{[h,d]}$  in $\Omega_k^L$ with positive incoming flow.
Starting from $t_k^{[h,d]}$, trace backward along positive-flow edges
until reaching $s_k^{[0,0]}$.
This produces an expanded path
$
\widehat{\pi}:
s_k^{[0,0]}\to\cdots\to t_k^{[h,d]}
$, with a bottleneck flow value $b(\widehat{\pi}) \triangleq \min_{(\mu,\nu)\in\widehat{\pi}} f_k(\mu,\nu) > 0$.
The backward trace always ends at $s_k^{[0,0]}$ due to both acyclicity and flow conservation.
Every backward step decreases the hop index, so the process cannot continue indefinitely.
If it stopped at a non-source state while the selected outgoing edge had positive flow,
flow conservation would imply positive incoming flow at that state, contradicting termination.
Therefore it must reach $s_k^{[0,0]}$.
Project the expanded path $\widehat{\pi}$ to obtain the corresponding physical path $\pi$ in $G$ .
Since deleting cycles does not violate the link-capacity constraints,
does not increase fidelity loss, and does not increase path hop count,
we can assume that $\pi$ is a simple path, without loss of generality.
If there exists $(\pi, \alpha) \in \cP_k$ for some $\alpha >0$,
replace $(\pi, \alpha)$ with $(\pi, \alpha + b(\widehat{\pi}) \times q^{h- 1})$.
Otherwise, set $\cP_k \leftarrow \cP_k \cup \{(\pi, b(\widehat{\pi}) \times q^{h- 1})\}$.
We then set $f_k(e) \leftarrow f_k(e)-b(\widehat{\pi})$ for every edge $e$
on $\widehat{\pi}$ to obtain the residual flow.
We continue this process until no terminal node has a positive incoming flow.

\begin{lemma}
\label{the:02}
\normalfont
An instance $\tiRS(G, \cD, L, \bT)$ is feasible if and only if the optimal value
$\Phi_L^\star$ of $\mathrm{LP}_{\mathrm{RS}}(G, \cD, L)$ satisfies
$\Phi_L^\star \ge \bT $.
Moreover, whenever $\Phi_L^\star\ge \bT$, a feasible rate
allocation $\cP = (\cP_1, \cP_2, \ldots, \cP_K)$ of $\tiRS(G, \cD, L, \bT)$ can be extracted from an optimal solution of
$\mathrm{LP}_{\mathrm{RS}}(G, \cD, L)$ using the above path extraction process.
\hfill$\Box$
\end{lemma}

\begin{proof}
Suppose first that the instance $\tiRS(G, \cD, L, \bT)$ is feasible.
Assume that 
$\cP = (\cP_1, \cP_2, \ldots, \cP_K)$
is a feasible rate allocation for $\tiRS(G, \cD, L, \bT)$.
We must have
$\Delta(\cP_k) \le L$ for $k \in [K]$,
$\sum_{k=1}^K R_k(\cP) \ge \bT$
and the link-capacity constraints in (\ref{eqn:new3.5})
are satisfied.

Let $(\pi, r_{\pi}) \in \cP_k$ for some $k \in [K]$.
Let
$\pi=v_0\to v_1\to\cdots\to v_h$,
and let
$
d_i=\sum_{a=1}^{i}\ell(v_{a-1},v_a), i = 1, 2, \ldots, h.
$
Define $d_0 = 0$ as a technical convention.
Since $d_h=\ell(\pi)\le L$, $\pi$ induces the expanded path
$
v_0^{[0,0]}
\to v_1^{[1,d_1]}
\to\cdots\to
v_h^{[h,d_h]}.
$
Define $f_k^{\pi}(v_{a-1}^{[a-1, d_{a - 1}]}, v_{a}^{[a, d_a]}) = \frac{r_{\pi}}{q^{h-1}}$ for $a=1, 2, \ldots, h$.
For any $k \in [K]$ and any $(\mu, \nu) \in E^{\cD}_L$, define
\begin{align}
\label{eqn:new5.4}
f_k (\mu, \nu) = \sum_{(\pi, r_\pi) \in \cP_k} f_k^{\pi}(\mu, \nu).
\end{align}
For each positive-rate physical path $\pi$, its induced expanded path
carries the same elementary-pair generation flow
$
r_\pi/q^{h-1}
$
on every expanded edge, and hence satisfies flow conservation at all
intermediate expanded nodes.
Therefore, summing over all paths preserves constraint~\eqref{lp:5.1b}.
Moreover, for any physical edge $e$, summing the constructed flows over
all expanded copies corresponding to $e$ gives exactly
$
\sum_{k=1}^K
\sum_{\substack{(\pi,r_\pi)\in\mathcal P_k\\ e\in\pi}}
\frac{r_\pi}{q^{h-1}}
$,
which is at most $c(e)$ by the link-capacity constraints in~\eqref{eqn:new3.5}.
Hence, constraint~\eqref{lp:5.1c} is also satisfied.
This shows that $f$ is a feasible solution for $\mathrm{LP}_{\mathrm{RS}}(G, \cD, L)$
with an objective function value equal to $\sum_{k=1}^K R_k(\cP) \ge \bT$.
Therefore
$\Phi_L^\star \ge \sum_{k=1}^K R_k(\cP) \ge \bT$.


Conversely, suppose $\mathrm{LP}_{\mathrm{RS}}(G, \cD, L)$ has a feasible solution with
objective value at least $\bT$.
Let $f$ be an optimal solution for $\mathrm{LP}_{\mathrm{RS}}(G, \cD, L)$ with an objective
function value $\Phi_L^\star \ge \bT$.
Apply the path-extraction procedure described above to obtain a rate allocation
$\cP = (\cP_1,\ldots,\cP_K)$.
In each iteration, subtracting the bottleneck value $b(\widehat{\pi})$
from every edge of $\widehat{\pi}$ preserves nonnegativity and flow conservation;
each extraction zeros at least one positive-flow expanded edge,
 so the process terminates.
%
Consider an extracted expanded path $\widehat{\pi}$ ending at $t_k^{[h,d]}$.
After projection and cycle removal, it yields a simple physical path
$\pi$ with
$
\ell(\pi)\le d\le L.
$
The assigned delivered rate is
$
r_\pi=q^{h-1}b(\widehat{\pi}).
$
If cycle removal reduces the number of hops from $h$ to $h'\le h$,
then the required elementary-pair generation rate on $\pi$ is
$
\frac{r_\pi}{q^{h'-1}}
=
q^{h-h'}b(\widehat{\pi})
\le
b(\widehat{\pi}),
$
so the link-capacity constraints remain satisfied.

Moreover, each extracted expanded path contributes
$
q^{h-1}b(\widehat{\pi})
$
to the LP objective and exactly the same amount to the delivered
throughput.
Therefore, for every commodity $k\in[K]$, summing over its extracted
paths gives
$
R_k(\mathcal P)=R_k(f).
$
Consequently,
$
\sum_{k=1}^K R_k(\mathcal P)
=
\sum_{k=1}^K R_k(f)
=
\Phi_L^\star.
$

Therefore, summing over all extracted paths gives a feasible
rate allocation
$
\mathcal P=(\cP_1,\ldots,\cP_K)
$
whose total delivered throughput equals $\Phi_L^\star \ge \bT$.
Hence, the iRS instance is feasible.
%
%
%
\end{proof}

We summarize our algorithm for iRS in Algorithm \ref{alg:pseudo_RSi}.



\begin{lemma}
\label{thm:pseudo_RSi_complexity}
Algorithm~1 correctly solves the iRS problem in
$
O\!\left((Kmn(L+1))^4\mathcal L\right)
$
time, where $\mathcal L$ denotes the input size of the iRS instance.
\end{lemma}

\begin{proof}
The correctness follows directly from Lemma~\ref{the:02}.
The expanded LP has
$
O(Kmn(L+1))
$
variables and
$
O(Kmn(L+1)+Kn^2(L+1)+m)
$
constraints.
Since the underlying undirected graph of $G$ is connected,
$m\ge n-1$, and hence both quantities are
$
O(Kmn(L+1)).
$
Using Ye's potential-reduction algorithm~\cite{ye1991n},
the LP can therefore be solved in
$
O\!\left((Kmn(L+1))^4\mathcal L\right)
$
time.

Path extraction takes
$
O(Kmn^2(L+1))
$
time, since at most $O(Kmn(L+1))$ paths are extracted and each can be
processed in $O(n)$ time.
This is dominated by the LP-solving time.
Thus, Algorithm~1 runs in
$
O\!\left((Kmn(L+1))^4\mathcal L\right)
$
time.
\end{proof}


\begin{algorithm}[htbp]
\caption{\textsc{Exact-iRS}$(G, \mathcal D, L, \bT)$}
\label{alg:pseudo_RSi}
\begin{algorithmic}[1]
\State Solve $\mathrm{LP}_{\mathrm{RS}}(G, \cD, L)$ and obtain  the optimal value $\Phi_L^\star$.
\If{$\Phi_L^\star<\bT$}
    \State \textbf{STOP}: infeasible.
\Else
    \State Extract and \Return $\mathcal P$.
\EndIf
\end{algorithmic}
\end{algorithm}

\noindent
\textbf{Implementation and graph pruning.}
In the implementation, the expanded graph is pruned separately for each
commodity before constructing the LP.
For commodity $k$, a forward BFS from the source state
$s_k^{[0,0]}$ identifies all expanded states reachable from the source.
A multi-source BFS from the terminal states in $\Omega_k^L$ on the
reversed expanded graph identifies all states from which at least one
terminal state is reachable.
We retain only the states contained in the intersection of these two
sets, together with the expanded edges between them.
All other states, their incident edges, and the corresponding flow
variables are removed.
For each commodity, the expanded graph contains
$O(n^2(L+1))$ states and $O(mn(L+1))$ edges.
Thus, the two searches take
$
O(n^2(L+1)+mn(L+1))
=
O(mn(L+1))
$
time.
Hence, pruning the expanded graphs for all $K$ commodities takes
$
O(Kmn(L+1))
$
time.
%
%
\section{An FPTAS for $\tRSo$}
\label{sec:scaling_based_approximation}
\noindent 
%
In this section, we develop an FPTAS for solving $\tRSo$.
For a minimization problem, an FPTAS returns, for any $\epsilon \in (0,1]$,
a feasible solution whose objective value is at most $1+\epsilon$
times the optimal value, in time polynomial in the input size and
$1/\epsilon$.
Accordingly, for oRS we seek a rate allocation $\mathcal P$
satisfying
$\Delta(\mathcal P)\le(1+\epsilon)\Delta^{\toRS}_{\mathrm{opt}}$.

We consider the case 
$\Delta^{\toRS}_{\mathrm{opt}} (G, \cD, \bT) = 0$
and the case
$\Delta^{\toRS}_{\mathrm{opt}} (G, \cD, \bT) > 0$
separately.
Let $G_0$ be the subgraph of $G$ with only the edges $e \in E$ such that
$\ell (e) = 0$.
Solve $\mathrm{LP}_{\mathrm{RS}} (G_0, \cD, 0)$ and denote its optimal objective function value by
$\Phi_0^\star$.
If $\Phi_0^\star \ge \bT$, then $\Delta^{\toRS}_{\mathrm{opt}} (G, \cD, \bT) = 0$
and an optimal solution to
$\toRS(G, \cD, \bT)$
can be obtained by the path extraction from the flow values $f_k(e)$.
If $\Phi_0^\star < \bT$, then $\Delta^{\toRS}_{\mathrm{opt}} (G, \cD, \bT) > 0$.
We will compute a $(1+\epsilon)$-approximation in this case.
Therefore, we assume $\Delta^{\toRS}_{\mathrm{opt}}>0$ in the rest of this section.
%
%
%

\subsection{Scaling and Rounding}
\noindent
We now consider the general case where the fidelity losses may be real-valued.
Let $\theta > 0$ be a scaling parameter.
For each edge $e\in E$, define the \textit{scaled loss}
$\ell_\theta(e)=\lfloor \theta\ell(e) \rfloor +1$.
The $+1$ is deliberate: it ensures that the scaled loss does not
underestimate $\theta\ell(e)$, while introducing at most one unit
of additive rounding error per edge.
We have 
\begin{align}
  \theta \ell(e) \le  \ell_\theta(e) \le \theta \ell(e) + 1.
  \label{eqn:6.1}
\end{align}
%
%
For a path $\pi$, let
$
\ell_\theta(\pi)=\sum_{e\in\pi}\ell_\theta(e).
$
Since any path in G has at most $n-1$ edges, we have
\begin{align}
    \theta\ell(\pi) \le \ell_\theta(\pi)
    \le \theta\ell(\pi)+n-1.
     \label{eqn:6.2}
\end{align}

For any integer threshold $L_\theta \ge 0$, let $\textbf{iRS}_\theta(G, \mathcal{D}, L_\theta, \bT)$ denote the instance with the fidelity loss  $\ell_\theta$ and the fidelity-loss budget  $L_\theta$.
We can use Algorithm~\ref{alg:pseudo_RSi} to test the feasibility of
this instance.
Let $L_\theta^\star$ be the smallest integer threshold such that
$\tRSi_\theta(G, \mathcal{D},  L_\theta, \bT)$ is feasible.

\subsection{Bounding the Optimal Value}
%
\noindent
We show how to compute $\tLB > 0$ and $\tUB > 0$ such that
$\tLB \le \Delta^{\mathrm{oRS}}_{opt} \le \tUB$
and $\tUB \le 4 \times \tLB$.
We call $\tLB$ and $\tUB$ a lower bound and an upper bound of $\Delta^{\mathrm{oRS}}_{opt}$,
respectively.


\noindent
\textbf{Initial Bounds}:
We first show how to compute a pair of $\tLB$ and $\tUB$ such that
$\tUB \le n\tLB$.
Let
$
0<\alpha_1<\alpha_2<\cdots<\alpha_J
$
be the distinct positive edge-loss values, and let $G(\alpha_j)$ be the
subgraph induced by the edges with loss at most $\alpha_j$.
For each $\alpha_j$, we test whether $G(\alpha_j)$ can support throughput
$\bT$ without a path fidelity-loss constraint.
To do so, we temporarily assign zero loss to every edge in $G(\alpha_j)$
and solve $\mathrm{LP}_{\mathrm{RS}}(G(\alpha_j), \cD, 0)$.
The resulting expanded graph retains the hop-count dimension, which is
still needed to account for the swapping success probability.
Since feasibility is monotone in $\alpha_j$, the smallest feasible value
$\alpha_{j^\star}$ can be found by solving $O(\log m)$ LPs using binary search.
We set
$
\tLB_0=\alpha_{j^\star}
$
and
$
\tUB_0=n\alpha_{j^\star}.
$
We have
\begin{align}
\tLB_0
\le
\Delta_{\mathrm{opt}}^{\mathrm{oRS}}
\le
(n-1)\alpha_{j^\star}
\le
\tUB_0
\le
n\tLB_0.
\label{eqn:6.3}
\end{align}
The lower bound follows from the minimality of $\alpha_{j^\star}$.
Indeed, if
$
\Delta_{\mathrm{opt}}^{\mathrm{oRS}}<\alpha_{j^\star},
$
every positive-rate path of an optimal allocation would use only edges
with losses strictly smaller than $\alpha_{j^\star}$, contradicting the
definition of $\alpha_{j^\star}$.
For the upper bound, feasibility in $G(\alpha_{j^\star})$ yields a flow
that can be decomposed into simple paths.
Each such path has at most $n-1$ edges, each of loss at most
$\alpha_{j^\star}$~\cite{misra2009polynomial,lorenz2001simple,xue2008polynomial}.

\noindent
\textbf{Approximate Testing}:
We use the approximate testing technique used in~\cite{lorenz2001simple,misra2009polynomial}
to tighten the bounds.
For any $\Delta>0$ and $\zeta\in(0,1]$, define
$\textsc{Test}(\Delta,\zeta)$ as follows.
Set
$
\theta=\frac{n-1}{\Delta\zeta}
$
and
$
L_\theta=\left\lfloor\frac{n-1}{\zeta}\right\rfloor+n-1.
$
Then $\textsc{Test}(\Delta,\zeta)$ returns YES if
$\tRSi_\theta(G,\mathcal D ,L_\theta, \bT)$ is feasible, and
returns NO otherwise.
The scaling factor is chosen so that the maximum additive rounding error
of a simple path, namely $n-1$, corresponds to at most
$\zeta\Delta$ in the original loss scale.
Thus, a scaled-feasible solution may exceed the tested threshold
$\Delta$ by at most the factor $1+\zeta$.
\begin{lemma}
\label{lem:test_property}
For any $\Delta>0$ and $\zeta\in(0,1]$, the following properties hold:
\begin{itemize}
    \item If $\textsc{Test}(\Delta,\zeta)=\textsc{Yes}$, then
    $
    \Delta^{\toRS}_{opt}\le (1+\zeta)\Delta;
    $
    \item if $\textsc{Test}(\Delta,\zeta)=\textsc{No}$, then
    $
    \Delta^{\toRS}_{opt}>\Delta.
    $
\end{itemize}

\hfill$\Box$
\end{lemma}

\begin{proof}
First, suppose that $\textsc{Test}(\Delta,\zeta)=\textsc{Yes}$. Then there
exists a feasible rate allocation $\mathcal P$ such that every
positive-rate path has scaled loss at most $L_\theta$. Hence, for every
positive-rate path $\pi$ in $\mathcal P$,
\begin{align}
\theta\ell(\pi)
\le
\ell_\theta(\pi)
\le
L_\theta
=
\left\lfloor\frac{n-1}{\zeta}\right\rfloor+n-1
\le
\frac{n-1}{\zeta}+n-1.
\label{eqn:6.4}
\end{align}
Dividing by $\theta= (n-1)/(\Delta\zeta)$ gives
$
\ell(\pi) \le \Delta(1+\zeta).
$
Therefore $\mathcal P$ is feasible for oRS with maximum fidelity loss
at most $(1+\zeta)\Delta$, and so
$
\Delta^{\toRS}_{opt}\le (1+\zeta)\Delta.
$

Next, suppose
$\Delta^{\mathrm{oRS}}_{opt}\le \Delta$, and let $\mathcal P^\star$ be an optimal rate
allocation for oRS.
Then every positive-rate path $\pi$ in
$\mathcal P^\star$ satisfies $\ell(\pi)\le \Delta$. By the scaling inequality,
\begin{align}
\ell_\theta(\pi) \le
\theta\ell(\pi)+n-1 \le
\theta\Delta+n-1
= \frac{n-1}{\zeta}+n-1.
 \label{eqn:6.5}
\end{align}
Since $\ell_\theta(\pi)$ is an integer, we have
$\ell_\theta(\pi)$$\le$ $\left\lfloor\frac{n-1}{\zeta}\right\rfloor$ $+n$$-$$1$$=L_\theta$.
Thus $\mathcal P^\star$ is feasible, so
$\textsc{Test}(\Delta,\zeta)=\textsc{Yes}$. Therefore, if
$\textsc{Test}(\Delta,\zeta)=\textsc{No}$, then $\Delta^{\toRS}_{opt}>\Delta$.
\end{proof}

\noindent
\textbf{Tighten the Bounds}:
Starting from $\tLB_0$ and $\tUB_0$, we tighten the bounds as follows.
While $\tUB_i>4\tLB_i$, set
$
\Delta\gets\sqrt{\tLB_i\tUB_i/2}.
$
If $\textsc{Test}(\Delta,1)=\textsc{Yes}$, set
$\tLB_{i+1}\gets\tLB_i$ and
$\tUB_{i+1}\gets2\Delta$.
Otherwise, set
$\tLB_{i+1}\gets\Delta$ and
$\tUB_{i+1}\gets\tUB_i$.
By Lemma~\ref{lem:test_property}, both updates preserve the validity of the
lower and upper bounds.

To analyze the number of iterations, define
$
\rho_i=\tUB_i/\tLB_i.
$
The choice
$
\Delta=\sqrt{\tLB_i\tUB_i/2}
$
makes the ratio between the updated upper and lower bounds the same
in both branches.
Indeed, if $\textsc{Test}(\Delta,1)=\textsc{Yes}$, then
$
\rho_{i+1}
=
\frac{2\Delta}{\tLB_i}
=
\sqrt{2\rho_i}.
$
Otherwise,
$
\rho_{i+1}
=
\frac{\tUB_i}{\Delta}
=
\sqrt{2\rho_i}.
$
Hence, in either case,
$
\rho_{i+1}=\sqrt{2\rho_i}.
$
Let
$
z_i=\log_2(\rho_i/2).
$
Then
$
z_{i+1}=z_i/2.
$
Since $\rho_0\le n$, we have $z_0=O(\log n)$.
Therefore, after $O(\log\log n)$ iterations, $z_i\le1$, which is
equivalent to $\rho_i\le4$.
Thus, the procedure terminates with
$
\tLB\le\Delta_{\mathrm{opt}}^{\mathrm{oRS}}
\le\tUB\le4\tLB.
$

\subsection{Fully Polynomial Time Approximation Scheme for $\tRSo$}
%
\noindent
Assume that we have computed $\tLB$ and $\tUB$ such that
$
\tLB\le \Delta_{\mathrm{opt}}^{\mathrm{oRS}}
\le \tUB\le 4\tLB.
$
For a given $\epsilon \in (0,1]$, set
$
\theta=\frac{n-1}{\epsilon\tLB}.
$
Then
$
\frac{n-1}{\theta}
=
\epsilon\tLB
\le
\epsilon\Delta_{\mathrm{opt}}^{\mathrm{oRS}}.
$
Thus, the maximum additive rounding error of a simple path is at most
an $\epsilon$-fraction of the optimal fidelity loss.
The constant-factor lower bound is therefore what allows the additive
rounding error introduced by scaling to be converted into a
multiplicative approximation error.
We can compute a rate allocation $\cP$ whose total throughput is at least
$\bT$ and whose maximum path fidelity loss is at most
$(1+\epsilon) \Delta^{\toRS}_{opt}$.
To summarize the above, the complete \textbf{FPTAS-oRS} is presented in
Algorithm~\ref{alg:fptas_ors}.

\begin{algorithm}[htbp]
\caption{\textsc{FPTAS-oRS}$(G,\mathcal D,\bT,\epsilon)$}
\label{alg:fptas_ors}
\small
\begin{algorithmic}[1]

\State Let $G_0$ be the subgraph of $G$ induced by the zero-loss edges.
Solve $\mathrm{LP}_{\mathrm{RS}}(G_0, \cD,0)$.
\If{$\Phi_0^\star\ge \bT$}
    \State Extract and \Return the corresponding rate allocation $\mathcal P$.
\EndIf

\State Compute $\alpha_{j^\star}$.
\If{$\alpha_{j^\star}$ does not exist}
    \State \textbf{STOP:} infeasible.
\EndIf

\State Set
$\tLB_0\gets\alpha_{j^\star}$,
$\tUB_0\gets n\alpha_{j^\star}$,
and $i\gets0$.

\While{$\tUB_i>4\tLB_i$}
    \State Set
    $
    \Delta\gets\sqrt{\tLB_i\tUB_i/2}.
    $
    \If{$\textsc{Test}(\Delta,1)=\textsc{Yes}$}
        \State
        $\tUB_{i+1}\gets2\Delta$,
        $\tLB_{i+1}\gets\tLB_i$.
    \Else
        \State
        $\tUB_{i+1}\gets\tUB_i$,
        $\tLB_{i+1}\gets\Delta$.
    \EndIf
    \State $i\gets i+1$.
\EndWhile

\State Set $\tLB\gets\tLB_i$ and $\tUB\gets\tUB_i$.
\State Set $\theta\gets\frac{n-1}{\epsilon\tLB}$
and $U_\theta\gets\lfloor\theta\tUB\rfloor+n-1.$

\State By bisection over the integer thresholds
$L_\theta\in[0,U_\theta]$,
find the minimum feasible threshold $L_\theta^\star$ for
$\tRSi_\theta(G,\mathcal D,L_\theta,\bT)$.

\State Solve
$\tRSi_\theta(G,\mathcal D,L_\theta^\star,\bT)$
and extract a rate allocation $\mathcal P$.

\State \Return $\mathcal P$.

\end{algorithmic}
\end{algorithm}

\begin{theorem}
\label{thm:fptas_ors}
Algorithm ~\ref{alg:fptas_ors} correctly reports infeasibility if
$oRS(G, \cD, \bT)$ is infeasible.
Otherwise, for any $\epsilon \in (0, 1]$, it returns a rate allocation $\cP$ that satisfies link-capacity constraints, $\sum_{k=1}^{K}R_k \ge \bT$ and 
$\Delta(\cP) \le (1+\epsilon) \Delta^{\toRS}_{opt}$.
%
%
Its running time is bounded by
$
O\!\left(
K^4m^4n^8
\left(
\log\log n+
\frac{1}{\epsilon^4}\log\frac{n}{\epsilon}
\right)
\mathcal L
\right),
$
where $\mathcal L$ is the input size of the 
$\tRSo$ instance.$\Box$
%
%
%
\end{theorem}
\begin{proof}
If the zero-loss test succeeds, Algorithm~\ref{alg:fptas_ors}
returns an allocation $\mathcal P$ with
$\Delta(\mathcal P)=\Delta^{\toRS}_{\mathrm{opt}}=0$,
and the approximation guarantee holds trivially.
Hence, in the remainder of the proof, assume
$\Delta^{\toRS}_{\mathrm{opt}}>0$.
If $\alpha_{j^\star}$ does not exist, then even
$G(\alpha_J)$, which contains all physical edges, cannot support total
throughput $\bT$ without any path fidelity-loss constraint.
Hence, the oRS instance is infeasible.
Conversely, if the oRS instance is feasible and
$\Delta_{\mathrm{opt}}^{\mathrm{oRS}}>0$, such an
$\alpha_{j^\star}$ must exist.

By the initial-bound construction and the bound-tightening procedure, we obtain
$\tLB \le \Delta^{\toRS}_{opt} \le \tUB \le 4 \tLB$.

\emph{Approximation guarantee:}
Algorithm~\ref{alg:fptas_ors} sets
$
\theta
=
\frac{n-1}{\epsilon \tLB}.
$
Let
$
U_\theta
=
\lfloor\theta \tUB\rfloor+n-1.
$
Consider an optimal rate allocation $\mathcal P^\star$ for $\tRSo$.
For every positive-rate path $\pi$ in $\cP ^\star$,
$
\ell(\pi)\le\Delta^{\toRS}_{opt}\le \tUB.
$
By the scaling inequality,
$
\ell_\theta(\pi)
\le
\theta\ell(\pi)+n-1
\le
\theta \tUB+n-1.
$
Since $\ell_\theta(\pi)$ is an integer,
$
\ell_\theta(\pi)
\le
\lfloor\theta \tUB\rfloor+n-1
=
U_\theta.
$
Thus, 
$\tRSi_\theta(G,\cD,U_\theta,\bT)$ is feasible.

Feasibility of $\tRSi_\theta(G,\cD,L,\bT)$ is monotone in $L$.
Hence, Algorithm~\ref{alg:fptas_ors} can compute the minimum feasible threshold $L_\theta^\star$ using bisection with the upper bound $U_\theta$.
For every positive-rate path $\pi$ in
$\mathcal P^\star$,
$
\ell_\theta(\pi) \le
\theta\ell(\pi)+n-1
\le \theta\Delta^{\toRS}_{opt}+n-1$.
Since $\ell_\theta(\pi)$ is an integer, $\cP ^\star$ is feasible for
the scaled instance with threshold
$
\lfloor\theta\Delta^{\toRS}_{opt}\rfloor+n-1.
$
By the minimality of $L_\theta^\star$,
$
L_\theta^\star
\le
\lfloor\theta\Delta^{\toRS}_{opt}\rfloor+n-1
\le
\theta\Delta^{\toRS}_{opt}+n-1.
$
Therefore,
$
\frac{L_\theta^\star}{\theta}
\le
\Delta^{\toRS}_{opt}+\frac{n-1}{\theta}
=
\Delta^{\toRS}_{opt}+\epsilon \tLB
\le
(1+\epsilon)\Delta^{\toRS}_{opt}.
$

Let $\mathcal P$ be the rate allocation extracted from
$
\tRSi_\theta
(G,\mathcal D,L_\theta^\star,\bT).
$
For every positive-rate path $\pi$ in $\mathcal P$,
$
\theta\ell(\pi)
\le
\ell_\theta(\pi)
\le
L_\theta^\star.
$
It follows that
\begin{align}
\Delta(\mathcal P)
\le
\frac{L_\theta^\star}{\theta}
\le
(1+\epsilon)\Delta^{\toRS}_{opt}.
 \label{eqn:6.6}
\end{align}

By Lemma~\ref{the:02}, the extracted allocation
also satisfies the throughput threshold and all link-capacity
constraints in~\eqref{eqn:new3.5}. Hence, $\mathcal P$ is a $(1+\epsilon)$-approximation to oRS.

\emph{Running time:}
$\alpha_{j^\star}$ can be found by bisection using $O(\log m)$ hop-expanded-flow LP solves, which is dominated by the subsequent phases.
For the bound-tightening procedure, each call to
$\textsc{Test}(\Delta,1)$ uses the scaled threshold
$O(n)$,
 by Lemma~\ref{thm:pseudo_RSi_complexity}, taking
$O\!\left( K^4m^4n^8\mathcal L\right) $ time.
The while-loop terminates after
$O(\log\log n)$ iterations.
Its total running time is therefore
$O\!\left( K^4m^4n^8 \log\log n\,
\mathcal L \right)$.

Since $\tUB/\tLB\le4$, we have
$U_\theta$$=$$\lfloor\theta 
\tUB\rfloor$$+$$n$$-$$1$$=O(n/\epsilon)$.
By monotonicity, $L_\theta^\star$ can be found by bisection using
$O(\log(n/\epsilon))$ scaled instances. By
Lemma~\ref{thm:pseudo_RSi_complexity}, the running time is
$
O\!\left(
K^4m^4n^8
\frac{1}{\epsilon^4}
\log\frac{n}{\epsilon}
\,\mathcal L
\right).
$
Combining all above, the total running time is bounded by
$O\!\left(K^4m^4n^8 \left( \log\log n+ \frac{1}{\epsilon^4}\log\frac{n}{\epsilon}
\right) \mathcal L \right).
$
\end{proof}


%
%
\section{FPTAS for $\tRSmRCo$}
\label{sec:throughput_fairness_algorithms}
\noindent 
In this section, we present the FPTAS scheme to solve oRS-mRC.

\subsection{Integer RS-mRC}
\label{subsec:ors_mrc_fptas}
\noindent
Let $\Delta_{\mathrm{opt}}^{\mathrm{mRC}}$ denote the optimal value of
$\tRSmRCo(G,\mathcal D,\bT,\bT_{\min})$.
We first consider the integer-loss version of RS-mRC.
Let
$\tRSmRCi(G,\mathcal D,L,\bT,\bT_{\min})$
denote an RS-mRC instance in which
$\ell(e)\in\mathbb Z_{\ge 0}$ for every $e\in E$ and $L$ is an integer
fidelity-loss threshold.
The problem asks whether there exists a rate allocation $\mathcal P$
that satisfies the link-capacity constraints,
$\Delta(\mathcal P)\le L,$
$\sum_{k=1}^{K}R_k\ge \bT, $
and
$R_k\ge \bT_{\min},  k=1,2,\ldots,K$.

For a fixed integer threshold $L$, we construct the hop-fidelity expanded
graph $G_L^{\mathcal D}$ as in Section~\ref{sec:integer_loss_exact}.
We then solve the following LP:
%
\begin{subequations}
\label{lp:7.1}
\begin{align}
\mathrm{LP}_{\mathrm{mRC}}(G, \cD, L, \bT_{\min}):
\quad
\text{max}\quad
&
\Phi_L^{\mathrm{mRC}}
=
\sum_{k=1}^{K}R_k(f)
\label{lp:7.1a}
\\
\text{s.t.}\quad
&
R_k(f)\ge \bT_{\min},
\qquad
k=1,2,\ldots,K,
\label{lp:7.1b}
\\
&
\text{Constraints~\eqref{lp:5.1b}--%
\eqref{lp:5.1e}.}
\label{lp:7.1c}
\end{align}
\end{subequations}
%
Here, $R_k(f)$ is the delivered throughput of commodity $k$, as defined
in~\eqref{eqn:5.2}.
Let $\bigl(\Phi_L^{\mathrm{mRC}}\bigr)^\star$ denote the optimal value of
$\mathrm{LP}_{\mathrm{mRC}}(G, \cD, L, \bT_{\min})$.

\begin{lemma}
\label{lem:integer_mrc_correctness}
\normalfont
An instance
$\tRSmRCi(G,\mathcal D,L,\bT,\bT_{\min})$
is feasible if and only if $\mathrm{LP}_{\mathrm{mRC}}(G, \cD, L, \bT_{\min})$ is feasible and
$
\bigl(\Phi_L^{\mathrm{mRC}}\bigr)^\star\ge \bT.
$
Moreover, whenever these conditions hold, a feasible rate allocation
can be extracted from an optimal LP solution.
%
\end{lemma}

\begin{proof}
By Lemma~\ref{the:02}, the expanded-flow
constraints characterize exactly the rate allocations whose
positive-rate paths have fidelity loss at most $L$ and that satisfy all
link-capacity constraints.
By the path-extraction result of Lemma~\ref{the:02}, the extracted rate allocation
satisfies
$
R_k(\mathcal P)=R_k(f)
$
for every $k\in[K]$.
Hence, constraints~\eqref{lp:7.1b} are preserved by the extraction
procedure, and every commodity satisfies
$
R_k(\mathcal P)\ge \bT_{\min}.
$
Therefore, such an allocation satisfies the RS-mRC requirements if and only if
the LP is feasible and its optimal total throughput satisfies
$
\bigl(\Phi_L^{\mathrm{mRC}}\bigr)^\star\ge\bT.
$
The rate allocation is obtained using the same flow-decomposition
procedure as in Lemma~\ref{the:02}.
\end{proof}

We summarize the exact algorithm for $\tRSmRCi$ in
Algorithm~\ref{alg:pseudo_RSi_mrc}.

\begin{algorithm}[htbp]
\caption{\textsc{Exact-iRS-mRC}%
$(G,\mathcal D,L,\bT,\bT_{\min})$}
\label{alg:pseudo_RSi_mrc}
\begin{algorithmic}[1]

\State Solve $\mathrm{LP}_{\mathrm{mRC}}(G, \cD, L, \bT_{\min})$.
\If{$\mathrm{LP}_{\mathrm{mRC}}(G, \cD, L, \bT_{\min})$ is infeasible}
    \State \textbf{STOP}: infeasible.
\EndIf
\State Let $\bigl(\Phi_L^{\mathrm{mRC}}\bigr)^\star$ be its optimal value.
\If{$\bigl(\Phi_L^{\mathrm{mRC}}\bigr)^\star<\bT$}
    \State \textbf{STOP}: infeasible.
\Else
    \State Extract and \Return $\mathcal P$.
\EndIf
\end{algorithmic}
\end{algorithm}

\subsection{FPTAS for oRS-mRC}
\noindent
With the exact algorithm for iRS-mRC, we can solve oRS-mRC approximately as follows.

\noindent
\textbf{Scaling and Rounding:}
\noindent
We use the same scaling and rounding scheme as in
Section~\ref{sec:scaling_based_approximation}.
For any integer scaled threshold $L_\theta\ge0$, let
$
\tRSmRCi_\theta
(G,\mathcal D,L_\theta,\bT,\bT_{\min})
$
denote the corresponding integer-loss instance under the scaled losses
$\ell_\theta$.
Its feasibility can be tested using Algorithm~\ref{alg:pseudo_RSi_mrc}.
Therefore, all scaling inequalities established for FPTAS-$\tRSo$
continue to hold without modification.

\noindent
We first test feasibility on the subgraph induced by the zero-loss edges,
using the same flow LP with the additional constraints
$
R_k(f)\ge\bT_{\min}
$
for all $k$.
If this LP is feasible and its optimal objective value is at least
$\bT$, then
$
\Delta_{\mathrm{opt}}^{\mathrm{mRC}}=0,
$
and an optimal rate allocation can be extracted.
Otherwise, either the oRS-mRC instance is infeasible or, if it is
feasible,
$
\Delta_{\mathrm{opt}}^{\mathrm{mRC}}>0.
$
The initial-bound search below distinguishes these two cases.

\noindent
\textbf{Initial Bounds:}
Let
$
0<\alpha_1<\alpha_2<\cdots<\alpha_J
$
be the distinct positive edge-loss values, and let $G(\alpha_j)$ be the
subgraph induced by the edges with loss at most $\alpha_j$.
For each $\alpha_j$, temporarily assign zero loss to every edge of
$G(\alpha_j)$ and solve
$
\mathrm{LP}_{\mathrm{mRC}}
(G(\alpha_j),\mathcal D,0,\bT_{\min}).
$
We say that $\alpha_j$ is feasible if this LP is feasible and its
optimal objective value is at least $\bT$.
Let $\alpha_{j^\star}$ be the smallest feasible value, which can be
found by binary search.
If $\alpha_{j^\star}$ does not exist, then the oRS-mRC instance is
infeasible even without a path fidelity-loss constraint.
Otherwise, since the zero-loss case has already been excluded,
$
\Delta_{\mathrm{opt}}^{\mathrm{mRC}}>0.
$

Set
$
\tLB_0=\alpha_{j^\star}
$
and
$
\tUB_0=n\alpha_{j^\star}.
$
By the minimality of $\alpha_{j^\star}$,
$
\alpha_{j^\star}
\le
\Delta_{\mathrm{opt}}^{\mathrm{mRC}}.
$
Conversely, feasibility of $G(\alpha_{j^\star})$ yields a rate
allocation satisfying the total-throughput and minimum-rate requirements
using simple paths with at most $n-1$ edges, each having loss at most
$\alpha_{j^\star}$.
Hence,
$
\Delta_{\mathrm{opt}}^{\mathrm{mRC}}
\le
(n-1)\alpha_{j^\star}.
$
Therefore,
\begin{align}
\tLB_0
=
\alpha_{j^\star}
\le
\Delta_{\mathrm{opt}}^{\mathrm{mRC}}
\le
(n-1)\alpha_{j^\star}
<
\tUB_0
=
n\alpha_{j^\star}
=
n\tLB_0.
\label{eqn:new7.2}
\end{align}

\noindent
\textbf{Approximate Testing}:
For any $\Delta>0$ and $\zeta\in(0,1]$, define
$\textsc{Test-mRC}(\Delta,\zeta)$ as follows.
Set
$\theta=\frac{n-1}{\Delta\zeta}$ and 
$L_\theta= \left\lfloor\frac{n-1}{\zeta}\right\rfloor+n-1.$
Then $\textsc{Test-mRC}(\Delta,\zeta)$ returns \textsc{Yes} if
$\tRSmRCi_\theta (G,\mathcal D,L_\theta,\bT,\bT_{\min})$
is feasible, and returns \textsc{No} otherwise.
The feasibility test is performed using
Algorithm~\ref{alg:pseudo_RSi_mrc}.

\begin{lemma}
\label{lem:test_mrc_property}
For any $\Delta>0$ and $\zeta\in(0,1]$, the following properties hold:
\begin{enumerate}
    \item If
    $\textsc{Test-mRC}(\Delta,\zeta)=\textsc{Yes}$, then
    $ \Delta_{\mathrm{opt}}^{\mathrm{mRC}} \le (1+\zeta)\Delta.$

    \item If
    $\textsc{Test-mRC}(\Delta,\zeta)=\textsc{No}$, then
    $\Delta_{\mathrm{opt}}^{\mathrm{mRC}}>\Delta.$
\end{enumerate}
\end{lemma}
\begin{proof}
The scaling inequalities~\eqref{eqn:6.1}--\eqref{eqn:6.2} remain unchanged.
If $\textsc{Test-mRC}(\Delta,\zeta)$ returns \textsc{Yes}, the extracted
allocation satisfies both the total-throughput and minimum-rate
requirements, while its original maximum path loss is at most
$(1+\zeta)\Delta$.
Conversely, if
$
\Delta_{\mathrm{opt}}^{\mathrm{mRC}}\le\Delta,
$
an optimal oRS-mRC allocation remains feasible under the scaled
threshold used by $\textsc{Test-mRC}$, by the same scaling argument as
in Lemma~\ref{lem:test_property}.
The additional minimum-rate constraints affect only the commodity rates
and therefore do not alter the loss-scaling argument.
\end{proof}


\noindent
\textbf{Tighten the Bounds}:
Starting from $\tLB_0$ and $\tUB_0$, we tighten the bounds using the same
procedure as for FPTAS-$\tRSo$.
While $\tUB_i>4\tLB_i$, set
$\Delta=\sqrt{\frac{\tLB_i\tUB_i}{2}}.$
If
$
\textsc{Test-mRC}(\Delta,1)=\textsc{Yes},
$
set
$\tLB_{i+1}=\tLB_i, \tUB_{i+1}=2\Delta.$
Otherwise, set
$\tLB_{i+1}=\Delta, \tUB_{i+1}=\tUB_i.$
By Lemma~\ref{lem:test_mrc_property}, all updated lower and upper bounds
remain valid.
After $O(\log\log n)$ iterations, we obtain
$\tLB \le \Delta_{\mathrm{opt}}^{\mathrm{mRC}}
\le \tUB \le 4\tLB.$

\noindent
\textbf{Final Scaling and Threshold Search}:
Assume that the bound-tightening procedure has produced
$\tLB$ and $\tUB$ satisfying
$\tLB \le \Delta_{\mathrm{opt}}^{\mathrm{mRC}}
\le \tUB \le 4\tLB.
$
For a given $\epsilon>0$, set
$\theta=\frac{n-1}{\epsilon\tLB}$
and
$U_\theta = \lfloor\theta\tUB\rfloor+n-1.$
Let $L_\theta^\star$ be the smallest integer threshold such that
$\tRSmRCi_\theta (G,\mathcal D,L_\theta^\star,\bT,\bT_{\min})
$
is feasible.

Such a threshold exists.
Indeed, let $\mathcal P^\star$ be an optimal rate allocation for
$\tRSmRCo$.
For every positive-rate path $\pi\in\mathcal P^\star$,
$\ell(\pi) \le \Delta_{\mathrm{opt}}^{\mathrm{mRC}}
\le \tUB.$
Therefore, by the scaling inequality,
$ \ell_\theta(\pi) \le \theta\ell(\pi)+n-1
\le \theta\tUB+n-1.$
Since $\ell_\theta(\pi)$ is an integer,
$\ell_\theta(\pi) \le \lfloor\theta\tUB\rfloor+n-1 = U_\theta.$
Hence,
$\tRSmRCi_\theta (G,\mathcal D,U_\theta,\bT,\bT_{\min})$
is feasible.
Since feasibility is monotone in the scaled threshold,
$L_\theta^\star$ can be found by bisection search.
We then solve the scaled instance at $L_\theta^\star$ and extract its
rate allocation.

We give the \textbf{FPTAS-oRS-mRC} in Algorithm~\ref{alg:fptas_ors_mrc}.

\begin{algorithm}[htbp]
\caption{\textsc{FPTAS-oRS-mRC}%
$(G,\mathcal D,\bT,\bT_{\min},\epsilon)$}
\label{alg:fptas_ors_mrc}
\begin{algorithmic}[1]

\State Let $G_0$ be the subgraph induced by the zero-loss edges.
\State Solve  $\mathrm{LP}_{\mathrm{mRC}}(G_0, \cD, 0, \bT_{\min})$.
\If{$\mathrm{LP}_{\mathrm{mRC}}(G_0, \cD, 0, \bT_{\min})$ is feasible and $(\Phi_0^{\mathrm{mRC}})^\star\ge\bT$}
    \State Extract and \Return the corresponding allocation $\mathcal P$.
\EndIf

\State Compute $\alpha_{j^\star}$.
\If{$\alpha_{j^\star}$ does not exist}
    \State \textbf{STOP}: infeasible.
\EndIf

\State Set
$
\tLB_0\gets\alpha_{j^\star},
$
$
\tUB_0\gets n\alpha_{j^\star},
$
and
$
i\gets0.
$

\While{$\tUB_i>4\tLB_i$}
    \State Set
    $
    \Delta\gets\sqrt{\tLB_i\tUB_i/2}.
    $
    \If{$\textsc{Test-mRC}(\Delta,1)=\textsc{Yes}$}
        \State Set
        $
        \tLB_{i+1}\gets\tLB_i
        $
        and
        $
        \tUB_{i+1}\gets2\Delta.
        $
    \Else
        \State Set
        $
        \tLB_{i+1}\gets\Delta
        $
        and
        $
        \tUB_{i+1}\gets\tUB_i.
        $
    \EndIf
    \State Set $i\gets i+1$.
\EndWhile

\State Set
$
\tLB\gets\tLB_i
$
and
$
\tUB\gets\tUB_i.
$

\State Set $\theta\gets\frac{n-1}{\epsilon\tLB}$
and $U_\theta\gets\lfloor\theta\tUB\rfloor+n-1.$

\State By bisection over the integer thresholds
$L_\theta\in[0,U_\theta]$,
find the minimum feasible threshold $L_\theta^\star$ for
$
\tRSmRCi_\theta
(G,\mathcal D,L_\theta,\bT,\bT_{\min})
$.

\State Solve the scaled instance at $L_\theta^\star$ and extract the
allocation $\mathcal P$.

\State \Return $\mathcal P$.

\end{algorithmic}
\end{algorithm}

\begin{theorem}
\label{the:ors_mrc_fptas}
Algorithm~\ref{alg:fptas_ors_mrc} correctly reports infeasibility if
the oRS-mRC instance is infeasible.
Otherwise, for any $\epsilon \in (0, 1]$, it returns a rate allocation
$\mathcal P$ satisfying the link-capacity constraints,
$
\sum_{k=1}^{K}R_k(\mathcal P)\ge\bT,
$
$
R_k(\mathcal P)\ge\bT_{\min},
\ k=1,2,\ldots,K,
$
and
$
\Delta(\mathcal P)
\le
(1+\epsilon)\Delta_{\mathrm{opt}}^{\mathrm{mRC}}.
$
Its running time is
$
O\!\left(
K^4m^4n^8
\left(
\log\log n
+
\frac{1}{\epsilon^4}\log\frac{n}{\epsilon}
\right)
\mathcal L
\right),
$
where $\mathcal L$ is the input size of the oRS-mRC instance.
\end{theorem}

\begin{proof}
If the zero-loss test succeeds, Algorithm~\ref{alg:fptas_ors_mrc}
returns a feasible rate allocation $\mathcal P$ with
$
\Delta(\mathcal P)
=
\Delta_{\mathrm{opt}}^{\mathrm{mRC}}
=
0,
$
and the approximation guarantee holds trivially.
If $\alpha_{j^\star}$ does not exist, then even the full physical graph
cannot satisfy the total-throughput and minimum-rate requirements
without any path fidelity-loss constraint.
Hence, the oRS-mRC instance is infeasible, and
Algorithm~\ref{alg:fptas_ors_mrc} correctly reports infeasibility.
Otherwise,
$
\Delta_{\mathrm{opt}}^{\mathrm{mRC}}>0.
$
Algorithm~\ref{alg:fptas_ors_mrc} then follows the same initial-bound,
bound-tightening, scaling, and final threshold-search arguments as
Algorithm~\ref{alg:fptas_ors}.
The additional constraints
$
R_k(f)\ge\bT_{\min},\ k\in[K],
$
affect only the commodity rates and do not alter the loss-scaling
arguments.
By Lemma~\ref{lem:integer_mrc_correctness}, the extracted allocation
satisfies the total-throughput and minimum-rate requirements.
Therefore, the same approximation argument as in
Theorem~\ref{thm:fptas_ors} gives
$
\Delta(\mathcal P)
\le
(1+\epsilon)\Delta_{\mathrm{opt}}^{\mathrm{mRC}}.
$
Hence,
$
\sum_{k=1}^{K}R_k(\mathcal P)\ge\bT
$
and
$
R_k(\mathcal P)\ge\bT_{\min}, k\in[K],
$
and all link-capacity constraints are satisfied.

\emph{Running time:}
Compared with $\mathrm{LP}_{\mathrm{RS}}(G, \cD, L)$, $\mathrm{LP}_{\mathrm{mRC}}(G, \cD, L, \bT_{\min})$ adds exactly
$K$ linear constraints and no flow variables.
Hence, Algorithm~3 has the same asymptotic running time as
Algorithm~1.
The number and sizes of the fixed-threshold oracle calls in the
zero-loss test, initial-bound search, bound-tightening phase, and final
bisection are asymptotically the same as those in
Algorithm~\ref{alg:fptas_ors}.
Therefore, the total running time is
$
O\!\left(
K^4m^4n^8
\left(
\log\log n
+
\frac{1}{\epsilon^4}\log\frac{n}{\epsilon}
\right)
\mathcal L
\right).
$
\end{proof}

\FloatBarrier

%
%
\section{An FPTAS for oMmFair}
\label{sec:FPTAS_MmFair}
\noindent
In this section, we present the FPTAS framework to  solve $\tMmFairo$.

\subsection{Integer MmFair}
\label{subsec:ommfair_extension}
\noindent
We first consider the fixed-threshold integer-loss problem
$\textbf{$\tMmFairi$}(G,\mathcal D,L,\boldsymbol{\eta})$,
where $\ell(e)\in\mathbb Z_{\ge 0}$ for every $e\in E$ and $L$ is an
integer fidelity-loss threshold.
The problem asks whether there exists a rate allocation
$\cP = (\cP_1, \cP_2, \ldots, \cP_K)$
satisfying the link-capacity constraints,
$\Delta(\mathcal P)\le L$,
and
$R(\cP)^\uparrow \succeq_{\mathrm{lex}} \boldsymbol{\eta}$.

By Lemma~\ref{the:02}, the expanded-flow constraints exactly characterize
the rate allocations satisfying the link-capacity constraints
and the fidelity-loss threshold $L$.
Therefore, to solve $\tMmFairi$, it suffices to compute the
lexicographically maximum sorted commodity-throughput vector achievable
under these expanded-flow constraints and compare it with
$\boldsymbol{\eta}$.
We compute this vector by progressively fixing bottleneck commodities,
following the leximin optimization approach
in~\cite{nace2008max,lan2010axiomatic}.

Let $\cQ$ be the set of commodities whose throughputs have already
been fixed, and let
$\cU$ be the set of unfixed commodities.
For each $k\in \cQ$, let $\gamma_k$ denote its fixed throughput.
Initially, $\cQ=\emptyset$ and $\cU=[K]$.
We first solve
%
\begin{subequations}
\label{lp:8.1}
\begin{align}
\mathrm{LP}_{\mathrm{fair}}(G, \cD, L, \mathcal Q):
\quad
\text{max}\quad
&
\lambda
\label{lp:8.1a}
\\
\text{s.t.}\quad
&
R_k(f)=\gamma_k,
\qquad
\forall k\in\mathcal Q,
\label{lp:8.1b}
\\
&
R_k(f)\ge\lambda,
\qquad
\forall k\in\mathcal U,
\label{lp:8.1c}
\\
&
\text{Constraints~\eqref{lp:5.1b}--%
\eqref{lp:5.1e}.}
\label{lp:8.1d}
\end{align}
\end{subequations}
%

Let $\lambda_{\mathrm{opt}}(\mathcal Q)$ denote its optimal value.
Thus, $\lambda_{\mathrm{opt}}(\mathcal Q)$ is the maximum common
throughput that can be guaranteed to all currently unfixed commodities,
while preserving the throughputs of the fixed commodities.

To determine which unfixed commodities must remain at this common
throughput, for every $k\in\mathcal U$, we solve
%
\begin{subequations}
\label{lp:8.2}
\begin{align}
\mathrm{LP}_{\mathrm{fair}}(G, \cD, L, \mathcal Q,k):
\quad
\text{max}\quad
&
R_k(f)
\label{lp:8.2a}
\\
\text{s.t.}\quad
&
R_j(f)=\gamma_j,
\qquad
\forall j\in\mathcal Q,
\label{lp:8.2b}
\\
&
R_j(f)\ge\lambda_{\mathrm{opt}}(\mathcal Q), \quad
\forall j\in\mathcal U,
\label{lp:8.2c}
\\
&
\text{Constraints~\eqref{lp:5.1b}--%
\eqref{lp:5.1e}.}
\label{lp:8.2d}
\end{align}
\end{subequations}
%


Let $\lambda_{opt}(\cQ, k)$ be the optimal value of $\mathrm{LP}_{\mathrm{fair}}(G, \cD, L, \mathcal Q,k)$.
If $\lambda_{opt}(\cQ, k) = \lambda_{opt}(\cQ)$, then commodity $k$ is a 
bottleneck commodity. 
We fix the throughput of this commodity at $\lambda_{opt}(\cQ)$ and add it to $\cQ$ (removing it from $\cU$). 
Since at least one commodity is fixed in each iteration, after at most $K$
iterations, $\mathcal U=\emptyset$ and all commodity throughputs have been
fixed.
We then obtain the lexicographically maximum feasible vector and extract the allocation $\cP$ from
the final flow solution of $\mathrm{LP}_{\mathrm{fair}}(G, \cD, L, \mathcal Q)$ in the last iteration.

We give Algorithm~\ref{alg:mmfair_oracle} to solve $\tMmFairi$ by computing the
lexicographically maximum commodity-throughput vector
$\mathbf R^{\uparrow}$ achievable under the integer loss threshold
$L$.
In each iteration, we define $\mathcal{B}$ as the bottleneck set:
$\mathcal B(\mathcal Q)=\{
k\in\mathcal U: \lambda_{\mathrm{opt}}(\mathcal Q,k) $
$\allowbreak =
\lambda_{\mathrm{opt}}(\mathcal Q)\}.
$

\begin{algorithm}[htbp]
\caption{\textsc{Exact-iMmFair}$(G,\mathcal D,L,\boldsymbol{\eta})$}
\label{alg:mmfair_oracle}
\small
\begin{algorithmic}[1]

\State Construct the expanded graph $G_L^{\mathcal D}$.
Set
$
\boldsymbol{\beta}\gets(),
$
$
\mathcal Q\gets\emptyset,
$
and
$
\mathcal U\gets\{1,2,\ldots,K\}.
$

\While{$\mathcal U\neq\emptyset$}

    \State Compute $\lambda_{\mathrm{opt}}(\mathcal Q)$ by solving
    $
    \mathrm{LP}_{\mathrm{fair}}(G,\mathcal D,L,\mathcal Q)
    $
    and obtain an optimal flow solution $f^\star$.

    \For{each $k\in\mathcal U$}
        \State Compute $\lambda_{\mathrm{opt}}(\mathcal Q,k)$ by solving
        $
        \mathrm{LP}_{\mathrm{fair}}(G,\mathcal D,L,\mathcal Q,k).
        $
    \EndFor

    \State Set $\mathcal B\gets\emptyset$.

    \For{each $k\in\mathcal U$}
        \If{$\lambda_{\mathrm{opt}}(\mathcal Q,k)
        =\lambda_{\mathrm{opt}}(\mathcal Q)$}
            \State $\mathcal B\gets\mathcal B\cup\{k\}$.
        \EndIf
    \EndFor

    \For{each $k\in\mathcal B$}
        \State
        $
        \gamma_k\gets\lambda_{\mathrm{opt}}(\mathcal Q).
        $
    \EndFor

    \State Append $|\mathcal B|$ copies of
    $
    \lambda_{\mathrm{opt}}(\mathcal Q)
    $
    to $\boldsymbol{\beta}$.

    \State Set
    $
    \mathcal Q\gets\mathcal Q\cup\mathcal B
    $
    and
    $
    \mathcal U\gets\mathcal U\setminus\mathcal B.
    $

\EndWhile

\If{$
    \boldsymbol{\beta}
    \prec_{\mathrm{lex}}
    \boldsymbol{\eta}
    $}
    \State \textbf{STOP}: infeasible.
\EndIf
    
\State Extract $\mathcal P$ from the flow solution $f^\star$ obtained
in the final iteration;
by Lemma~\ref{the:02}, $R(\cP) ^\uparrow = \beta$.
\State \Return $\cP$.

\end{algorithmic}
\end{algorithm}

\begin{lemma}
\label{lem:mmfair_oracle_correctness}
\normalfont
For any fixed integer fidelity-loss threshold $L$,
Algorithm~\ref{alg:mmfair_oracle} computes the lexicographically maximum
sorted feasible commodity-throughput vector
$\boldsymbol{\beta}$.
Consequently, it returns a rate allocation $\cP$ if and only if
$
\tMmFairi(G,\mathcal D,L,\boldsymbol{\eta})
$
is feasible.
Whenever it returns $\cP$, the extracted rate allocation
$\mathcal P$ satisfies
$
\Delta(\mathcal P)\le L
$ and $
R(\mathcal P)^\uparrow
=
\boldsymbol{\beta}
\succeq_{\mathrm{lex}}
\boldsymbol{\eta}.
$
Moreover, the algorithm terminates after at most $K$ iterations and
solves $O(K^2)$ expanded-flow LPs.
\hfill$\Box$
\end{lemma}

\begin{proof}
Consider an iteration with fixed commodity set $\mathcal Q$ and unfixed
commodity set $\mathcal U$.
By definition,
$\lambda_{\mathrm{opt}}(\mathcal Q)$ is the largest common throughput
that can be simultaneously guaranteed to all commodities in
$\mathcal U$, while preserving
$
R_k(f)=\gamma_k
$
for every $k\in\mathcal Q$.

We first show that the bottleneck set $\mathcal B$ is nonempty.
Suppose, to the contrary, that
$
\lambda_{\mathrm{opt}}(\mathcal Q,k)
>
\lambda_{\mathrm{opt}}(\mathcal Q),
\forall k\in\mathcal U.
$
For each $k\in\mathcal U$, let $f^{(k)}$ be an optimal solution of
$
\mathrm{LP}_{\mathrm{fair}}
(G,\mathcal D,L,\mathcal Q,k).
$
Every such solution satisfies
$
R_j(f^{(k)})
\ge
\lambda_{\mathrm{opt}}(\mathcal Q),
\forall j\in\mathcal U,
$
and additionally
$
R_k(f^{(k)})
>
\lambda_{\mathrm{opt}}(\mathcal Q).
$
Since all constraints are linear, the average
$
\bar f
=
\frac{1}{|\mathcal U|}
\sum_{k\in\mathcal U} f^{(k)}
$
is feasible and satisfies
$
R_j(\bar f)
>
\lambda_{\mathrm{opt}}(\mathcal Q),
\forall j\in\mathcal U.
$
This contradicts the optimality of
$\lambda_{\mathrm{opt}}(\mathcal Q)$ in
$
\mathrm{LP}_{\mathrm{fair}}
(G,\mathcal D,L,\mathcal Q).
$
Hence,
$
\mathcal B\neq\emptyset.
$

For every $k\in\mathcal B$,
$
\lambda_{\mathrm{opt}}(\mathcal Q,k)
=
\lambda_{\mathrm{opt}}(\mathcal Q).
$
Thus, among all feasible solutions satisfying
$
R_j(f)
\ge
\lambda_{\mathrm{opt}}(\mathcal Q),
\forall j\in\mathcal U,
$
commodity $k$ cannot have throughput strictly greater than
$\lambda_{\mathrm{opt}}(\mathcal Q)$.
Therefore, every such solution satisfies
$
R_k(f)
=
\lambda_{\mathrm{opt}}(\mathcal Q),
\forall k\in\mathcal B.
$
Hence, all commodities in $\mathcal B$ can be fixed simultaneously at
$\lambda_{\mathrm{opt}}(\mathcal Q)$.

Moreover, an optimal solution of
$
\mathrm{LP}_{\mathrm{fair}}
(G,\mathcal D,L,\mathcal Q)
$
remains feasible after this update.
Hence, the common throughput computed in the next iteration cannot be
smaller than the current value, and the values appended to
$\boldsymbol{\beta}$ are nondecreasing.

We next establish lexicographic optimality.
At each iteration,
$\lambda_{\mathrm{opt}}(\mathcal Q)$ is the largest value that all
currently unfixed commodities can simultaneously attain.
Thus, no feasible throughput vector agreeing with the previously fixed
entries can have its next smallest entry larger than
$\lambda_{\mathrm{opt}}(\mathcal Q)$.
Furthermore, every commodity in $\mathcal B$ must equal this value
before any remaining commodity can be increased.
Therefore, the commodities fixed in the current iteration form the next
block of the lexicographically maximum sorted feasible throughput vector.
Applying this argument inductively over all iterations shows that the
final vector $\boldsymbol{\beta}$ is the lexicographically maximum sorted
feasible throughput vector.

In the final iteration, we must have
$
\mathcal B=\mathcal U,
$
since the update
$
\mathcal U\gets\mathcal U\setminus\mathcal B
$
makes $\mathcal U$ empty.
Hence, in the final optimal flow solution $f^\star$, every commodity
fixed in an earlier iteration has throughput equal to its prescribed
value $\gamma_k$, while every commodity remaining in the final
iteration has throughput
$
\lambda_{\mathrm{opt}}(\mathcal Q).
$
Therefore, sorting
$
R_1(f^\star),R_2(f^\star),\ldots,R_K(f^\star)
$
in nondecreasing order gives
$
\boldsymbol{\beta}.
$
By the path-extraction result of Lemma~\ref{the:02},
$
R_k(\mathcal P)
=
R_k(f^\star),
k\in[K].
$
Hence,
$
R(\mathcal P)^\uparrow
=
\boldsymbol{\beta}.
$

Since $\boldsymbol{\beta}$ is the lexicographically maximum sorted
feasible throughput vector, the iMmFair instance is feasible if and
only if
$
\boldsymbol{\beta}
\succeq_{\mathrm{lex}}
\boldsymbol{\eta}.
$
If
$
\boldsymbol{\beta}
\prec_{\mathrm{lex}}
\boldsymbol{\eta},
$
Algorithm~\ref{alg:mmfair_oracle} correctly reports infeasibility.
Otherwise, the extracted allocation $\mathcal P$ satisfies
$
R(\mathcal P)^\uparrow
=
\boldsymbol{\beta}
\succeq_{\mathrm{lex}}
\boldsymbol{\eta}.
$
By Lemma~\ref{the:02}, $\mathcal P$ also satisfies the link-capacity
constraints and
$
\Delta(\mathcal P)\le L.
$

Finally, since $\mathcal B\neq\emptyset$ in every iteration, at least one
commodity is fixed per iteration, so the algorithm terminates after at
most $K$ iterations.
With $|\mathcal U|$ unfixed commodities, one instance of
$
\mathrm{LP}_{\mathrm{fair}}(G,\mathcal D,L,\mathcal Q)
$
and $|\mathcal U|$ instances of
$
\mathrm{LP}_{\mathrm{fair}}(G,\mathcal D,L,\mathcal Q,k)
$
are solved.
Therefore, the total number of expanded-flow LPs is
$
K+\sum_{i=1}^{K} i
=
O(K^2).
$
\end{proof}
\subsection{FPTAS for oMmFair}
\noindent
Algorithm~\ref{alg:mmfair_oracle} serves as an exact
feasibility test for $\tMmFairi$ under any fixed integer-loss threshold.
We now use this exact test in the same scaling framework as
FPTAS-$\tRSo$ to solve $\tMmFairo$.
%
%
%
%

\noindent
\textbf{Scaling and Rounding:}
We use the same scaling and rounding scheme as in
Section~\ref{sec:scaling_based_approximation}.
For any integer scaled threshold $L_\theta\ge0$, let
$
\tMmFairi_\theta
(G,\mathcal D,L_\theta,\boldsymbol{\eta})
$
denote the corresponding integer-loss instance under the scaled losses
$\ell_\theta$.
Its feasibility can be tested using
Algorithm~\ref{alg:mmfair_oracle}.
Therefore, all scaling inequalities established for FPTAS-$\tRSo$
continue to hold without modification.

We first test the zero-loss case.
Let $G_0$ be the subgraph induced by the edges with zero loss, and run
$
\textsc{Exact-iMmFair}(G_0,\mathcal D,0,\boldsymbol{\eta}).
$
If it returns a rate allocation $\mathcal P$, then
$
\Delta_{\mathrm{opt}}^{\mathrm{fair}}=0,
$
and $\mathcal P$ is optimal.
Otherwise, either the oMmFair instance is infeasible or, if it is
feasible,
$
\Delta_{\mathrm{opt}}^{\mathrm{fair}}>0.
$
The initial-bound search below distinguishes these two cases.

\noindent
\textbf{Initial Bounds:}
Let
$
0<\alpha_1<\alpha_2<\cdots<\alpha_J
$
be the distinct positive edge-loss values, and let $G(\alpha_j)$ be the
subgraph induced by the edges whose losses are at most $\alpha_j$.
For each $\alpha_j$, temporarily assign zero loss to every edge of
$G(\alpha_j)$ and run
$
\textsc{Exact-iMmFair}
(G(\alpha_j),\mathcal D,0,\boldsymbol{\eta}).
$
We say that $\alpha_j$ is feasible if the algorithm returns a rate
allocation.
Let $\alpha_{j^\star}$ be the smallest feasible value, which can be
found by binary search.

If $\alpha_{j^\star}$ does not exist, then
$
\mathrm{oMmFair}(G,\mathcal D,\boldsymbol{\eta})
$
is infeasible even without a path fidelity-loss constraint.
Otherwise, since the zero-loss case has already been excluded,
$
\Delta_{\mathrm{opt}}^{\mathrm{fair}}>0.
$

Set
$
\tLB_0=\alpha_{j^\star}
$
and
$
\tUB_0=n\alpha_{j^\star}.
$
By the minimality of $\alpha_{j^\star}$,
$
\alpha_{j^\star}
\le
\Delta_{\mathrm{opt}}^{\mathrm{fair}}.
$
Conversely, feasibility of $G(\alpha_{j^\star})$ yields a rate
allocation
$
\mathcal P=(\mathcal P_1,\mathcal P_2,\ldots,\mathcal P_K)
$
satisfying
$
R(\mathcal P)^\uparrow
\succeq_{\mathrm{lex}}
\boldsymbol{\eta},
$
using simple paths with at most $n-1$ edges, each having loss at most
$\alpha_{j^\star}$.
Hence,
$
\Delta_{\mathrm{opt}}^{\mathrm{fair}}
\le
(n-1)\alpha_{j^\star}.
$
Therefore,
\begin{align}
\tLB_0
=
\alpha_{j^\star}
\le
\Delta_{\mathrm{opt}}^{\mathrm{fair}}
\le
(n-1)\alpha_{j^\star}
<
\tUB_0
=
n\alpha_{j^\star}
=
n\tLB_0.
\end{align}

\noindent
\textbf{Approximate Testing:}
For any $\Delta>0$ and $\zeta\in(0,1]$, define
$\textsc{Fair-Test}(\Delta,\zeta)$ as follows.
Set
$
\theta=\frac{n-1}{\Delta\zeta}$ and 
$L_\theta= \left\lfloor\frac{n-1}{\zeta}\right\rfloor+n-1.$
Then apply Algorithm~\ref{alg:mmfair_oracle} to
$
\tMmFairi_\theta
(G,\mathcal D,L_\theta,\boldsymbol{\eta}).
$
The procedure returns \textsc{Yes} if the algorithm returns a rate
allocation, and returns \textsc{No} otherwise.

\begin{lemma}
\label{lem:fair_test}
For any $\Delta>0$ and $\zeta\in(0,1]$,
$\textsc{Fair-Test}(\Delta,\zeta)$ satisfies the following:
if it returns \textsc{Yes}, then
$
\Delta_{\mathrm{opt}}^{\mathrm{fair}}
\le
(1+\zeta)\Delta;
$
if it returns \textsc{No}, then
$
\Delta_{\mathrm{opt}}^{\mathrm{fair}}
>
\Delta.
$
\end{lemma}

\begin{proof}
The scaling inequalities~(6.1)--(6.2) remain unchanged.
If $\textsc{Fair-Test}(\Delta,\zeta)$ returns \textsc{Yes},
Algorithm~\ref{alg:mmfair_oracle} finds a feasible scaled solution
satisfying the fairness threshold.
By Lemma~\ref{the:02}, path extraction preserves every commodity throughput,
and hence preserves the lexicographic fairness condition.
The same scaling argument as in Lemma~\ref{lem:test_property} then gives
$
\Delta_{\mathrm{opt}}^{\mathrm{fair}}
\le
(1+\zeta)\Delta.
$

Conversely, suppose
$
\Delta_{\mathrm{opt}}^{\mathrm{fair}}\le\Delta.
$
An optimal oMmFair allocation is feasible under loss threshold
$\Delta$, and by the scaling inequalities, every positive-rate path
in this allocation is feasible under the scaled threshold used by
$\textsc{Fair-Test}(\Delta,\zeta)$.
Since Algorithm~\ref{alg:mmfair_oracle} is an exact fixed-threshold
fairness test, $\textsc{Fair-Test}(\Delta,\zeta)$ must return
\textsc{Yes}.
Therefore, if it returns \textsc{No}, we must have
$
\Delta_{\mathrm{opt}}^{\mathrm{fair}}>\Delta.
$
\end{proof}

\noindent
\textbf{Tighten the Bounds}:
Starting from these bounds, we use
$\textsc{Fair-Test}(\Delta,1)$ to perform the same bound-tightening
procedure as in Algorithm~\ref{alg:fptas_ors}.
After $O(\log\log n)$ iterations, we obtain
$\tLB \le \Delta_{\mathrm{opt}}^{\mathrm{fair}}
\le \tUB \le 4\tLB. $

\noindent
\textbf{Final Scaling and Threshold Search:}
After the bound-tightening procedure, we have
$
\tLB
\le
\Delta_{\mathrm{opt}}^{\mathrm{fair}}
\le
\tUB
\le
4\tLB.
$
For a given $\epsilon>0$, set
$
\theta=\frac{n-1}{\epsilon\tLB}
$
and
$
U_\theta=\lfloor\theta\tUB\rfloor+n-1.
$
As in Section~6, since
$
\Delta_{\mathrm{opt}}^{\mathrm{fair}}\le\tUB,
$
the scaled version of an optimal oMmFair allocation is feasible under
the threshold $U_\theta$.
Hence, the minimum feasible scaled threshold can be found by bisection
over $[0,U_\theta]$.

The complete algorithm  \textbf{FPTAS-oMmFair} is given in
Algorithm~\ref{alg:fptas_ommfair}.

\begin{algorithm}[htbp]
\caption{\textsc{FPTAS-oMmFair}%
$(G,\mathcal D,\boldsymbol{\eta},\epsilon)$}
\label{alg:fptas_ommfair}
\small
\begin{algorithmic}[1]

\State Let $G_0$ be the subgraph induced by the zero-loss edges.
\State Run
$
\textsc{Exact-iMmFair}(G_0,\mathcal D,0,\boldsymbol{\eta}).
$
\If{it returns a rate allocation $\mathcal P$}
    \State \Return $\mathcal P$.
\EndIf

\State Compute the smallest feasible edge-loss threshold
$\alpha_{j^\star}$.
\If{$\alpha_{j^\star}$ does not exist}
    \State \textbf{STOP}: infeasible.
\EndIf

\State Set
$
\tLB_0\gets\alpha_{j^\star},
$
$
\tUB_0\gets n\alpha_{j^\star},
$
and
$
i\gets0.
$

\While{$\tUB_i>4\tLB_i$}
    \State Set
    $
    \Delta\gets\sqrt{\tLB_i\tUB_i/2}.
    $
    \If{$\textsc{Fair-Test}(\Delta,1)=\textsc{Yes}$}
        \State Set
        $
        \tLB_{i+1}\gets\tLB_i
        $
        and
        $
        \tUB_{i+1}\gets2\Delta.
        $
    \Else
        \State Set
        $
        \tLB_{i+1}\gets\Delta
        $
        and
        $
        \tUB_{i+1}\gets\tUB_i.
        $
    \EndIf
    \State Set $i\gets i+1$.
\EndWhile

\State Set
$
\tLB\gets\tLB_i
$
and
$
\tUB\gets\tUB_i.
$

\State Set
$
\theta\gets\frac{n-1}{\epsilon\tLB}
$
and
$
U_\theta\gets\lfloor\theta\tUB\rfloor+n-1.
$

\State By bisection over the integer thresholds
$
L_\theta\in[0,U_\theta],
$
find the minimum feasible threshold $L_\theta^\star$ for
$
\mathrm{iMmFair}_\theta
(G,\mathcal D,L_\theta,\boldsymbol{\eta}).
$

\State Apply Algorithm~\ref{alg:mmfair_oracle} to
$
\tMmFairi_\theta
(G,\mathcal D,L_\theta^\star,\boldsymbol{\eta})
$
and obtain its rate allocation $\mathcal P$.

\State \Return $\mathcal P$.

\end{algorithmic}
\end{algorithm}
\begin{theorem}
\label{thm:ommfair_fptas}
Algorithm~\ref{alg:fptas_ommfair} correctly reports infeasibility if
the oMmFair instance is infeasible.
Otherwise, for any $\epsilon \in (0, 1]$, it returns a rate allocation
$\mathcal P$ satisfying the link-capacity constraints,
$
R(\mathcal P)^\uparrow
\succeq_{\mathrm{lex}}
\boldsymbol{\eta},
$
and
$
\Delta(\mathcal P)
\le
(1+\epsilon)\Delta_{\mathrm{opt}}^{\mathrm{fair}}.
$
The running time is
$
O\!\bigl(
K^6m^4n^8
\bigl(
\log\log n
+
\epsilon^{-4}\log(n/\epsilon)
\bigr)
\mathcal L
\bigr),
$
where $\mathcal L$ is the input size of the oMmFair instance.
\end{theorem}
\begin{proof}
If the zero-loss test succeeds,
Algorithm~\ref{alg:fptas_ommfair} returns a feasible rate allocation
$\mathcal P$ with
$
\Delta(\mathcal P)
=
\Delta_{\mathrm{opt}}^{\mathrm{fair}}
=
0,
$
and the approximation guarantee holds trivially.
If $\alpha_{j^\star}$ does not exist, then even the full physical graph
cannot satisfy the fairness requirement
$
R(\mathcal P)^\uparrow
\succeq_{\mathrm{lex}}
\boldsymbol{\eta}
$
without any path fidelity-loss constraint.
Hence, the oMmFair instance is infeasible, and
Algorithm~\ref{alg:fptas_ommfair} correctly reports infeasibility.
Otherwise,
$
\Delta_{\mathrm{opt}}^{\mathrm{fair}}>0.
$
By Lemma~\ref{lem:mmfair_oracle_correctness},
Algorithm~\ref{alg:mmfair_oracle} is an exact feasibility oracle for
iMmFair under any fixed integer-loss threshold.
By Lemma~\ref{lem:fair_test}, the same scaling and bound-tightening
arguments as in Theorem~\ref{thm:fptas_ors} apply.
Therefore, the returned allocation satisfies
$
R(\mathcal P)^\uparrow
\succeq_{\mathrm{lex}}
\boldsymbol{\eta}
$
and
$
\Delta(\mathcal P)
\le
(1+\epsilon)\Delta_{\mathrm{opt}}^{\mathrm{fair}},
$
together with all link-capacity constraints.

\emph{Running time:}
Each expanded-flow LP solved by
Algorithm~\ref{alg:mmfair_oracle} has the same asymptotic size as the
fixed-threshold LP for iRS, and
Algorithm~\ref{alg:mmfair_oracle} requires $O(K^2)$ such LP solves.
The number and sizes of the fixed-threshold oracle calls in the
zero-loss test, initial-bound search, bound-tightening phase, and final
bisection are asymptotically the same as those in
Algorithm~\ref{alg:fptas_ors}.
Hence, the total running time is
$
O\!\bigl(
K^6m^4n^8
\bigl(
\log\log n
+
\epsilon^{-4}\log(n/\epsilon)
\bigr)
\mathcal L
\bigr).
$
\end{proof}

The proposed framework can also be used in a more general setting where
different commodities may have different fidelity requirements.
This extension is discussed in Appendix B.

%
%
\section{Performance Evaluation}
\label{esc:evaluation}
\noindent
We evaluate the proposed schemes on a machine running
Ubuntu~26.04 LTS, equipped with a 12th Gen Intel Core
i9-12900 processor (16 cores, up to
5.1\,GHz) and 128\,GB RAM.
All LPs were solved using Gurobi
Optimizer~13.0.2.
For each network size $n$, we generate five Erd\H{o}s--R\'enyi network topologies with expected degree 6.
Each undirected physical link is independently assigned a capacity uniformly sampled from $[20,200]$ and a Werner parameter uniformly sampled from $[0.9,1]$.
The swapping success probability is set as $q=0.8$.
For each topology, we generate 10 independent demand sets, each containing $K=4$ random commodities.

\begin{figure}[htbp]
    \centering
    \includegraphics[width=0.6\linewidth, height=6cm]{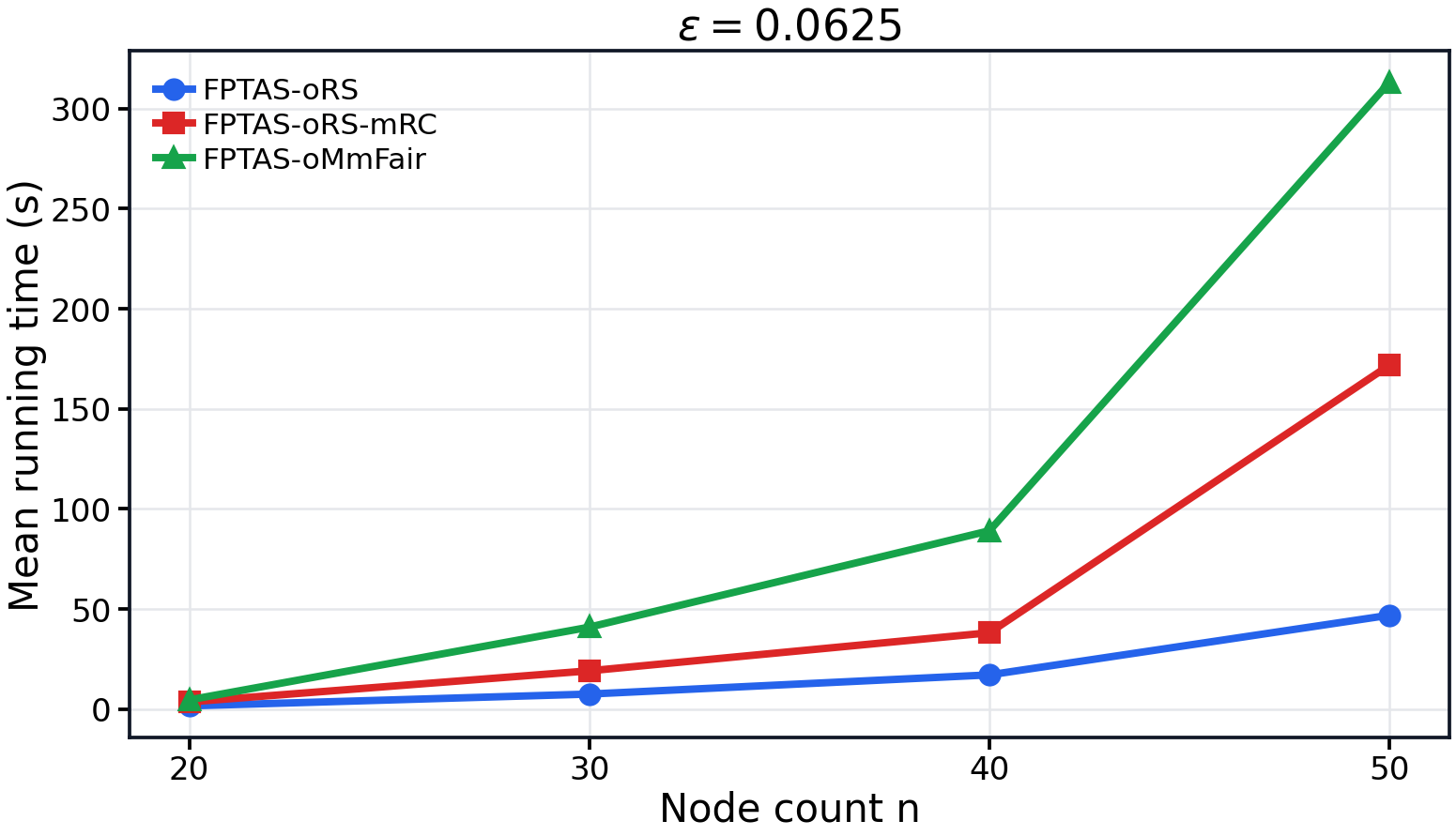}
    \caption{Mean running time versus the number of nodes for FPTAS algorithms with $\epsilon=0.0625$.}
    \label{fig:result_node}
\end{figure}

\begin{figure*}[htbp]
    \centering
    \includegraphics[width=1.0\linewidth, height=4.5cm]{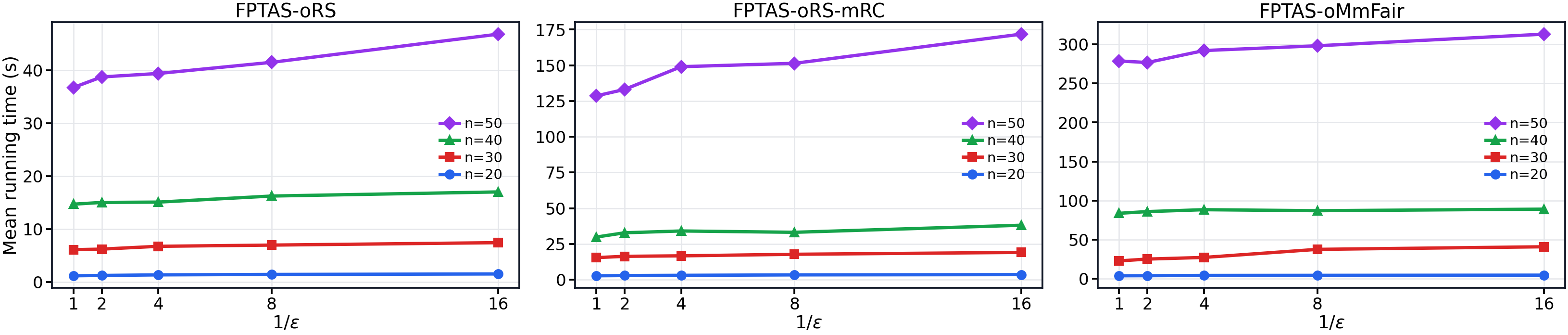}
    \caption{Mean running time versus $1/\epsilon$ for different network sizes.
}
    \label{fig:result_eps}
\end{figure*}

We test three objectives: oRS, oRS-mRC, and oMmFair.
The approximation parameter is varied over $\epsilon \in \{1,0.5,0.25,0.125,\allowbreak 0.0625\}$.
The target rate is set to $\bT=0.5T_{\text{max}}$, where $T_{\text{max}}$ is the maximum total throughput achievable without fidelity constraints.
For oRS-mRC, $\bT_{\text{min}}$ is set to $0.5T/K$.
For oMmFair, we set the fairness-threshold vector to
$\boldsymbol{\eta} =\frac{\bT}{K}(0.60,0.65,0.70,0.75)$.

Fig.~\ref{fig:result_node} shows the mean running time as the network size increases from $n=20$ to $50$ under the most accurate tested setting, $\epsilon=0.0625$. 
FPTAS-oRS completes within 50 s when $n=50$, whereas FPTAS-oMmFair requires  more time.
FPTAS-oRS-mRC generally requires more computation time than FPTAS-oRS,
as its feasibility LPs include additional per-commodity minimum-rate
constraints.

Fig.~\ref{fig:result_eps} reports the mean running time as a function of $1/\epsilon$ for different network sizes.
The running time tends to increase with $1/\epsilon$ but less than  proportionally over the tested range.
The reason is that as $1/\epsilon$ increases, the scaled expanded graph contains more candidate states, while many additional states are pruned: expanded states that are unreachable from the source or cannot reach the destination are removed, together with their incident expanded edges and flow variables.
\section{Conclusion}
\label{sec:conclusion}
\noindent
In this paper, we studied multi-pair rate allocation in quantum networks
subject to both throughput and fidelity considerations, where the links may
have different fidelity values.
We studied three problems: rate sum, rate sum subject to minimum-rate constraints,
and max-min fairness.
Our work extends~\cite{chakraborty2020entanglement}, which studies
rate sum under the assumption that all links have the same fidelity.
We proved that all three problems are NP-hard, even for a single-pair case.
We then studied optimization versions of the problems, and designed fully
polynomial-time approximation schemes for each of them.
While our approximation schemes build on the scaling-and-rounding
and approximate-testing techniques used in earlier papers~\cite{lorenz2001simple,xue2008polynomial},
we made nontrivial extensions of these techniques to capture the properties
of quantum networks.
We implemented our approximation schemes and tested them on randomly generated test cases.
Our evaluation results demonstrate the effectiveness of our schemes.

\appendices

\section*{Appendix A: Proof of Lemma~\ref{lem:new01}}
\setcounter{section}{1}
\setcounter{equation}{0}
\label{app:discussion}
%
\noindent
%
%
%
%
%
%
\textit{Proof.}
From (\ref{eqn:new4.7}), we have
%
\begin{align}
\label{eqn:A.1}
W_{\mathrm{opt}} = {\ee}^{- \Delta_{\mathrm{opt}}},~~
F_{\mathrm{opt}} = \frac{1 + 3 W_{\mathrm{opt}}}{4} \ge \frac{1}{4}.
\end{align}
%
%
From (\ref{eqn:new4.8}), we have
%
\begin{align}
\label{eqn:A.2}
W(\cP) = {\ee}^{- \Delta(\cP)},~~
F(\cP) = \frac{1 + 3 W(\cP)}{4}.
\end{align}
%
%
Therefore, inequality (\ref{eqn:new4.9}) implies
%
\begin{align}
\label{eqn:A.3}
W_{\mathrm{opt}} - W(\mathcal P)
& =
{\ee}^{- \Delta_{\mathrm{opt}}} - {\ee}^{- \Delta(\cP)}
\\
\label{eqn:A.4}
& \le
{\ee}^{- \Delta_{\mathrm{opt}}}
-
{\ee}^{-(1+\epsilon)\Delta_{\mathrm{opt}}}
\\
\label{eqn:A.5}
& =
{\ee}^{-\Delta_{\mathrm{opt}}}
\left(
1-{\ee}^{-\epsilon\Delta_{\mathrm{opt}}}
\right)
\\
\label{eqn:A.6}
& \le
\epsilon\Delta_{\mathrm{opt}}
{\ee}^{-\Delta_{\mathrm{opt}}}
\\
\label{eqn:A.7}
& \le
\frac{\epsilon}{\ee},
\end{align}
where
(\ref{eqn:A.4}) follows from (\ref{eqn:new4.9}),
(\ref{eqn:A.6}) follows from
$1 - {\ee}^{-x} \le x$ for $x \ge 0$,
and
(\ref{eqn:A.7}) follows from
$x{\ee}^{-x} \le 1/\ee$ for $x \ge 0$.

Therefore, we have
%
\begin{align}
F_{\mathrm{opt}} - F(\cP)
& =
\frac{1 + 3 W_{\mathrm{opt}}}{4} - \frac{1 + 3 W(\cP)}{4} 
\\
&=
\frac{3}{4}
\left(
W_{\mathrm{opt}}-W(\mathcal P)
\right)
\\
& \le
\frac{3\epsilon}{4 \ee}.
\end{align}
%
Moreover, $F_{\mathrm{opt}} \ge 1/4$.
Therefore,
%
\begin{align}
F_{\mathrm{opt}} - F(\mathcal P)
\le
\frac{3\epsilon}{\ee}F_{\mathrm{opt}}
\le
\delta F_{\mathrm{opt}}.
\end{align}
%
Therefore, we have
%
\begin{align}
F(\mathcal P)
\ge
(1-\delta)F_{\mathrm{opt}}.
\label{eqn:A.old5}
\end{align}
%
This completes the proof of Lemma~\ref{lem:new01}.
\hfill$\Box$


\setcounter{section}{2}
\renewcommand{\thesection}{\Alph{section}}

%
%
\setcounter{equation}{0}
\renewcommand{\theequation}{B.\arabic{equation}}
\section*{Appendix B: Extension}
\setcounter{section}{2}
\addcontentsline{toc}{section}{Appendix B. Discussion}
\label{app:extension}
\noindent
The decision problems RS, RS-mRC, and MmFair can be extended to the
setting where different commodities have different fidelity
requirements.
This extension also covers the rate-sum setting of
Chakraborty et al.~\cite{chakraborty2020entanglement}, in which
different commodities may have different end-to-end fidelity
requirements while the physical links have identical fidelity.
Hence, their rate-sum setting can be viewed as a special case of the
commodity-specific framework considered here.

Suppose commodity $k$ has a prescribed Werner-parameter threshold
$\bW_k$.
Let
$
\bW_0=\min_{k\in[K]}\bW_k
$
be a common reference threshold.
For each commodity $k$, introduce a new virtual source $s'_k$ and a
virtual edge
$
e'_k=(s'_k,s_k)
$
with Werner parameter
$
\widetilde w(e'_k)=\frac{\bW_0}{\bW_k}.
$
Since
$
\bW_0\le\bW_k,
$
we have
$
\widetilde w(e'_k)\in(0,1].
$
The original demand $(s_k,t_k)$ is replaced by $(s'_k,t_k)$.

Let
$
\widetilde G
=
(\widetilde V,\widetilde E,\widetilde c,\widetilde w,q)
$
denote the augmented network, where
$
\widetilde V
=
V\cup\{s'_1,s'_2,\ldots,s'_K\}
$
and
$
\widetilde E
=
E\cup\{e'_1,e'_2,\ldots,e'_K\}.
$
For every physical edge $e\in E$, set
$
\widetilde c(e)=c(e)
$
and
$
\widetilde w(e)=w(e).
$
For every virtual edge $e'_k$, set
$
\widetilde c(e'_k)=\sum_{e\in E}c(e).
$
Every unit of flow traversing $e'_k$ must subsequently traverse at least
one physical edge, so the flow on $e'_k$ cannot exceed the sum of the
loads on all physical edges, which is at most
$
\sum_{e\in E}c(e).
$
Thus, the virtual edge does not create an additional bottleneck.
Moreover, this capacity has a binary representation of polynomial length
in the input size.

The augmented commodity set is
$
\widetilde{\mathcal D}
=
\{(s'_k,t_k):k\in[K]\}.
$
For any $s_k$--$t_k$ path $\pi$ in the original network, let
$
\widetilde\pi_k=e'_k\circ\pi
$
denote the corresponding $s'_k$--$t_k$ path in the augmented network.
Then
\begin{align}
\widetilde W(\widetilde\pi_k)=
\widetilde w(e'_k)W(\pi)=
\frac{\bW_0}{\bW_k}W(\pi).
\end{align}
Therefore,
\begin{align}
\widetilde W(\widetilde\pi_k)\ge\bW_0
\quad\Longleftrightarrow\quad
W(\pi)\ge\bW_k.
\label{eqn:B.1}
\end{align}
Equivalently, let
$
\Delta_k=-\ln\bW_k
$
and
$
\Delta_0=-\ln\bW_0.
$
The loss of the virtual edge is
\begin{align}
\widetilde\ell(e'_k) =
-\ln\frac{\bW_0}{\bW_k} =
\Delta_0-\Delta_k.
\end{align}
Hence,
\begin{align}
\widetilde\ell(\widetilde\pi_k)
&=
\ell(\pi)+\Delta_0-\Delta_k,
\end{align}
and consequently,
\begin{align}
\widetilde\ell(\widetilde\pi_k)\le\Delta_0
\quad\Longleftrightarrow\quad
\ell(\pi)\le\Delta_k.
\label{eqn:B.2}
\end{align}
Thus, the commodity-specific fidelity requirements in the original
network are equivalent to the common fidelity requirement $\bW_0$ in
the augmented network.

\medskip
\noindent
\textbf{Integer-Loss Formulation for RS:}
We next describe the corresponding fixed-threshold formulation for the
integer-loss version of RS.
Suppose the physical edge losses and
$
\Delta_k,\ k\in[K],
$
are nonnegative integers.
Then
$
\widetilde\ell(e'_k)=\Delta_0-\Delta_k
$
is also a nonnegative integer.
For the transformed decision instance, set
$
L=\Delta_0.
$

An original simple path contains at most $n-1$ physical edges.
Therefore, its corresponding augmented path consists of one virtual
edge followed by at most $n-1$ physical edges and hence contains at
most $n$ edges.
Accordingly, construct the augmented hop-fidelity expanded graph
$
\widetilde G_L^{\widetilde{\mathcal D}}
=
(\widetilde V_L^{\widetilde{\mathcal D}},
 \widetilde E_L^{\widetilde{\mathcal D}})
$
with
\begin{align}
\widetilde V_L^{\widetilde{\mathcal D}}
=
\left\{
u^{[h,d]}:
u\in\widetilde V,\;
0\le h\le n,\;
0\le d\le L
\right\},
\end{align}
and
\begin{align}
\widetilde E_L^{\widetilde{\mathcal D}}
=
\Bigl\{
&
\bigl(
u^{[h,d]},
v^{[h+1,d+\widetilde\ell(e)]}
\bigr):
e=(u,v)\in\widetilde E,
0\le h<n,\;
0\le d\le L-\widetilde\ell(e)
\Bigr\}.
\end{align}

For commodity $k$, define the terminal-state set
$
\widetilde\Omega_k^L
=
\left\{
t_k^{[h,d]}:
2\le h\le n,\;
0\le d\le L
\right\}.
$
The lower bound $h\ge2$ follows because every augmented path contains
one virtual edge followed by at least one physical edge.
The virtual edge is introduced only for the transformation and does not
represent an additional physical swapping operation.
Hence, its additional hop must be compensated when computing the
delivered throughput.
For each commodity $k$, define
\begin{align}
\widetilde R_k(f)
=
q^{-1}
\sum_{h=2}^{n}
\sum_{d=0}^{L}
q^{h-1}
\sum_{\mathclap{
(\mu,t_k^{[h,d]})
\in
\widetilde E_L^{\widetilde{\mathcal D}}
}}
f_k(\mu,t_k^{[h,d]}).
\label{eqn:newB.8}
\end{align}
Indeed, if an original physical path has $r$ hops and carries
elementary-pair generation flow $x$, its delivered rate is
$
q^{r-1}x.
$
The corresponding augmented path has $r+1$ hops and would contribute
$
q^r x
$
under the standard expanded-flow objective.
The factor $q^{-1}$ restores the original delivered rate, since
$
q^{-1}q^r x=q^{r-1}x.
$

The fixed-threshold integer-loss formulation for the
commodity-specific version of RS is
\begin{subequations}
\label{lp:B.4}
\begin{align}
\text{max}\quad
&
\widetilde\Phi_L
=
\sum_{k=1}^{K}\widetilde R_k(f)
\label{lp:B.4a}
\\[1mm]
\text{s.t.}\quad
&
\sum_{(\mu,\nu)\in
\widetilde E_L^{\widetilde{\mathcal D}}}
f_k(\mu,\nu)
=
\sum_{(\nu,\xi)\in
\widetilde E_L^{\widetilde{\mathcal D}}}
f_k(\nu,\xi),
\nonumber\\[-1mm]
&
\hspace{2em}
\forall
\nu\in
\widetilde V_L^{\widetilde{\mathcal D}}
\setminus
\left(
\{(s'_k)^{[0,0]}\}
\cup
\widetilde\Omega_k^L
\right),
\quad
\forall k\in[K],
\label{lp:B.4b}
\\[2mm]
&
\sum_{k=1}^{K}
\sum_{h=0}^{n-1}
\sum_{d=0}^{L-\widetilde\ell(e)}
f_k\!\left(
u^{[h,d]},
v^{[h+1,d+\widetilde\ell(e)]}
\right)
\le
\widetilde c(e),
\nonumber\\[-1mm]
&
\hspace{2em}
\forall e=(u,v)\in\widetilde E:
\widetilde\ell(e)\le L,
\label{lp:B.4c}
\\[2mm]
&
f_k(\mu,\nu)=0,
\quad
\forall(\mu,\nu)\in
\widetilde E_L^{\widetilde{\mathcal D}},
\quad
\forall\mu\in\widetilde\Omega_k^L,
\quad
\forall k\in[K],
\label{lp:B.4d}
\\
&
f_k(\mu,\nu)\ge0,
\quad
\forall k\in[K],
\quad
\forall(\mu,\nu)\in
\widetilde E_L^{\widetilde{\mathcal D}}.
\label{lp:B.4e}
\end{align}
\end{subequations}

By the same argument as in Lemma~\ref{the:02}, together with the compensated
throughput definition in~\eqref{eqn:newB.8}, the commodity-specific RS
instance is feasible if and only if
$
\widetilde\Phi_L^\star\ge\bT.
$
Thus, LP~\eqref{lp:B.4} plays the same role for the
commodity-specific version of RS as LP~\eqref{lp:5.1}
does for the common-threshold version.

\medskip
\noindent
\textbf{Path Extraction:}
The same path-decomposition procedure as in Section~5 can be used, with
an adjustment to the delivered rate.
Consider an extracted augmented expanded path ending at
$
t_k^{[h,d]}
$
with bottleneck flow value $b(\widehat\pi)$.
After dropping the hop-loss indices and removing the first virtual edge,
we obtain an $s_k$--$t_k$ physical walk in the original network.
After cycle removal, let $\pi$ denote the resulting simple physical
path.

The delivered rate assigned to $\pi$ is
$r_\pi = q^{-1}q^{h-1}b(\widehat\pi)
= q^{h-2}b(\widehat\pi).$
This is exactly the throughput contribution of the augmented path in
\eqref{eqn:newB.8}.
Before cycle removal, the corresponding physical walk has $h-1$
physical edges.
If cycle removal reduces its hop count to $h'\le h-1$, then the
elementary-pair generation rate required on each edge of the resulting
physical path is
$\frac{r_\pi}{q^{h'-1}}=
q^{h-h'-1}b(\widehat\pi) \le b(\widehat\pi)$,
because $q\in(0,1]$ and $h-h'-1\ge0$.
Hence, projection and cycle removal do not violate the 
link-capacity constraints.
Furthermore, removing the virtual edge and deleting cycles cannot
increase the physical path loss, so the corresponding
commodity-specific fidelity requirement is also preserved.

\medskip
\noindent
\textbf{RS-mRC and MmFair:}
For RS-mRC, we use LP~\eqref{lp:B.4} with the additional constraints
$\widetilde R_k(f)\ge\bT_{\min},
k\in[K].$
The commodity-specific RS-mRC instance is feasible if this modified LP
is feasible and its optimal objective value is at least $\bT$.
For MmFair, LP~\eqref{lp:B.4} itself is not used as the fairness
objective.
Instead, we use the same augmented expanded-flow constraints and the
compensated commodity throughputs $\widetilde R_k(f)$ in the fairness
formulations of Section~8.
Specifically, every occurrence of $R_k(f)$ in
$
\mathrm{LP}_{\mathrm{fair}}
(\widetilde G,\widetilde{\mathcal D},L,\mathcal Q)
$
and
$
\mathrm{LP}_{\mathrm{fair}}
(\widetilde G,\widetilde{\mathcal D},L,\mathcal Q,k)
$
is replaced by
$
\widetilde R_k(f).
$
Algorithm~\ref{alg:mmfair_oracle} can then be applied to these modified
LPs to test whether the commodity-specific MmFair instance satisfies
the prescribed fairness threshold.

Therefore, the commodity-specific-threshold versions of RS, RS-mRC,
and MmFair can be reduced to their corresponding common-threshold
fixed-loss formulations on the augmented network.

We emphasize that the optimization-problem instances studied in this
paper do not include commodity-specific fidelity thresholds as input.
Thus, the above construction extends the decision problems, rather than
directly extending the FPTASs for oRS, oRS-mRC, and oMmFair.
Extending the optimization problems would require a generalized
objective and a corresponding approximation analysis.
Nevertheless, this transformation provides a useful building block for
future studies of more general resource-allocation problems with
commodity-specific fidelity requirements.
\bibliographystyle{IEEEtranS}
\bibliography{refs}
\end{document}